\documentclass[letterpaper,onecolumn,accepted=2026-01-06]{quantumarticle}
\pdfoutput=1
\usepackage{url}
\usepackage{graphicx}
\usepackage{authblk}
\usepackage{amssymb,amsmath,amsthm,dsfont,bbm}
\usepackage{multicol}
\usepackage{titlesec}
\usepackage{color}
\usepackage{mathtools}
\usepackage{thm-restate}
\usepackage{tikz-cd}
\usepackage{soul}
\usepackage{xcolor,colortbl}
\usepackage[colorlinks=true,linkcolor=blue,citecolor=blue]{hyperref}

\usetikzlibrary{positioning}

\usepackage{braket,complexity,hyperref,cleveref}

\usepackage[margin=1.in]{geometry}
\usepackage{cleveref}

\usepackage{algorithm}
\usepackage[noend]{algpseudocode}

\renewcommand{\hat}{\widehat}
\renewcommand{\tilde}{\widetilde}

\newcommand{\eps}{\varepsilon}

\newcommand{\Dd}{\mathcal{D}}

\newcommand{\Nn}{\mathcal{N}}

\newcommand{\Uu}{\mathcal{U}}

\newcommand{\trnc}{\text{trnc}}
\newcommand{\trunc}{\operatorname{trnc}\hspace{0.05em}}
\newcommand{\sqfor}{\mathsf{Squared Forrelation}}
\newcommand{\distbal}{\mathsf{Balance Checking}}
\newcommand{\bllqsv}{\text{Black-Box }\mathsf{LLQSV}}
\newcommand{\llqsv}{\mathsf{LLQSV}}

\DeclareMathOperator{\negl}{negl}
\DeclareMathOperator{\sgn}{sgn}
\DeclareMathOperator{\erfc}{erfc}

\DeclareMathOperator{\RejSamp}{RejSamp}
\let\E\relax
\newcommand{\E}{\mathop{\bf E\/}}
\renewcommand{\Pr}{\mathop{\bf Pr\/}}
\newcommand{\Cov}{\mathop{\bf Cov\/}}

\renewcommand{\poly}{\mathrm{poly}}

\newclass{\QCAM}{QCAM}

\newtheorem{theorem}{Theorem}
\newtheorem{lemma}[theorem]{Lemma}
\newtheorem{claim}[theorem]{Claim}

\newtheorem{corollary}[theorem]{Corollary}
\newtheorem{definition}[theorem]{Definition}
\newtheorem{conjecture}[theorem]{Conjecture}
\newtheorem{fact}[theorem]{Fact}

\newtheorem{remark}[theorem]{Remark}

\newtheorem{problem}[theorem]{Problem}

\begin{document}
\title{On Certified Randomness from Fourier Sampling or Random Circuit Sampling}

\author{Roozbeh Bassirian}\affiliation{University of Chicago}
\email{roozbeh@uchicago.edu}
\author{Adam Bouland}\affiliation{Stanford University}
\email{abouland@stanford.edu}
\author{Bill Fefferman}\affiliation{University of Chicago}
\email{wjf@uchicago.edu}
\author{Sam Gunn}\affiliation{UC Berkeley}
\email{gunn@berkeley.edu}
\author{Avishay Tal}\affiliation{UC Berkeley}
\email{atal@berkeley.edu}
\maketitle
\begin{abstract}
Certified randomness has a long history in quantum information, with many potential applications. Recently, Aaronson (2018, 2020) proposed a novel public certified randomness protocol based on existing random circuit sampling (RCS) experiments.
The security of his protocol, however, relies on non-standard complexity-theoretic conjectures which were not previously studied in the literature.

Inspired by Aaronson's work, we study certified randomness in the quantum random oracle model (QROM). 
We show that quantum Fourier Sampling can be used to define a publicly verifiable certified randomness protocol with black-box security without any computational assumptions.
In addition to giving a certified randomness protocol in the QROM, our work can also be seen as supporting Aaronson's conjectures for RCS-based randomness generation, as our protocol is in some sense the ``black-box version'' of Aaronson's protocol.
In further support of Aaronson's proposal, we prove a Fourier Sampling version of Aaronson's conjecture by extending Raz and Tal's separation of $\BQP$ vs $\PH$.

Our work complements the subsequent certified randomness protocol of Yamakawa and Zhandry (2022) in the QROM. 
Whereas the security of that protocol relied on the Aaronson-Ambainis conjecture, ours does not rely on any computational assumption --- at the expense of requiring exponential-time classical verification. Our protocol also has a simple heuristic implementation.
\end{abstract}

\thispagestyle{empty}
\clearpage
\pagebreak
\setcounter{page}{1}
\newpage

 \section{Introduction}

The task of producing certified randomness -- i.e. numbers that can be proved random to a skeptic, without any trust in the device  -- is a uniquely quantum ability, with a long history in quantum information.
Certified randomness has many possible applications, for example in the certification of public lotteries, the selection of cryptographic keys, or proof of work cryptocurrencies (see e.g. \cite{acinsurvey} for a discussion).
It was first realized that Bell inequality violations could enable certified randomness using multiple non-communicating devices \cite{Pironio2010}.
Subsequently, single-device certified randomness schemes have been designed that use interactive protocols and assume post-quantum cryptography \cite{BCM+18, BKVV20, KCVY21, HG21, LG21}.
Recently, in unpublished work, Aaronson \cite{a-talk1,a-talk} has suggested that it might even be possible to generate public certified randomness with a single device using near-term random circuit sampling (RCS) experiments.
Aaronson’s proposal therefore offers the possibility that recent quantum advantage experiments \cite{boixo-supremacy,ustc-rcs} might have useful applications.

The security of Aaronson's protocol, however, relies on a custom complexity-theoretic conjecture called the “long list quantum supremacy verification assumption," or $\llqsv$.
This conjecture is both strong and nonstandard -- it essentially conjectures that the output distribution of RCS is not only $\textsf{BQP}$-indistinguishable from random, but is indistinguishable for quantum variants of $\textsf{AM}$, even with access to exponentially many samples.
It is therefore important to critically examine its validity, and more generally offer evidence that statistical tests for RCS certify the generation of entropy.
More broadly, Aaronson's work opens the question of what other certified randomness protocols might be possible in the single-device, non-interactive setting -- and whether or not it is possible to achieve certified randomness under weaker or different assumptions.

\vspace{-1em}
\subsection{Our results in context}

Our first result studies certified randomness in the context of Fourier Sampling in the quantum random oracle model (QROM).
Here the task is, given black-box access to a random function $f$, output a sample from $f$'s Fourier spectrum.
One can easily achieve this by querying $f$ on a uniform superposition and measuring in the Hadamard basis.
We show that Fourier Sampling can be used to generate certified randomness.
We prove that in this black-box setting, any quantum device passing a certain statistical test based on Fourier Sampling\footnote{ The test checks that the outputs are heavy Fourier coefficients, similar to $\mathsf{FourierFishing}$ \cite{aaronsonbqpph} or Heavy Output Generation (HOG) \cite{ac}.} must generate a high amount of min-entropy.
We then use this test to construct a cryptographic protocol to generate certified randomness, non-interactively using a single device.

Our second result provides black-box evidence for Aaronson's $\llqsv$ conjecture. Informally, $\llqsv$ asks to distinguish whether an (exponentially long) list of random quantum circuits and outcomes $(C_i, s_i)$ are uncorrelated, or $s_i$ is sampled from the output distribution of $C_i$. \emph{Black-Box} $\llqsv$ replaces random circuits with random Boolean functions, and the output distribution is defined by the Fourier coefficients of $f_i$. We can prove that black-box $\llqsv$ is not solvable in $\BQP$ or $\PH$. Even though this is an oracle result, it can be seen as an attempt to come closer to attaining certified randomness in the white box setting, since the $\llqsv$ conjecture was designed for this purpose~\cite{a-talk}.

We believe these results are interesting for the following reasons.
First, they provide complexity-theoretic support for Aaronson's certified randomness protocol.
As we will discuss, the intellectual history of RCS is intimately tied with analogous developments in Fourier sampling (e.g. \cite{aaronsonbqpph}).
Our Fourier sampling protocol is in some sense the ``black-box'' version of Aaronson's protocol -- and therefore, our protocol's security is evidence in favor of the security of Aaronson's construction. 
To further elucidate this point, we prove a variant of Aaronson's $\llqsv$ conjecture in the Fourier sampling domain, which requires a nontrivial extension of the recent $\BQP$ vs $\PH$ oracle separation of Raz and Tal \cite{DBLP:conf/stoc/RazT19} to the $\llqsv$ setting. This provides evidence to back Aaronson's conjectures.

Second, our work provides a certified randomness protocol in its own right.
Fourier Sampling is a well-studied computational primitive in the black-box model \cite{aaronsonbqpph} and a plausible candidate for a quantum supremacy experiment in a white-box model \cite{fefferman2015power}.
Our work shows this experiment is capable of generating certified randomness in the black-box setting, and it is plausible that it might be instantiable in the white-box setting via sufficiently strong cryptographic hash functions.

Third, our work complements a subsequent certified randomness protocol of Yamakawa and Zhandry \cite{YK22}. Assuming the Aaronson-Ambainis conjecture~\cite{AA14}, Yamakawa and Zhandry designed a certified randomness protocol in the QROM model which is both non-interactive and \emph{efficiently verifiable} -- i.e. the tests performed on the output can be done in polynomial time.
Their work was the first to demonstrate that efficient verification of certified randomness is in principle possible using a non-interactive protocol.
Our work, however, offers two advantages over Yamakawa and Zhandry's protocol.
First, our protocol is closer to experimental feasibility, as their protocol not only requires instantiating the QROM, but also performing the list decoding of folded Reed-Solomon codes in superposition.
Second, while the security of Yamakawa and Zhandry's protocol relies on the Aaronson-Ambainis conjecture, our result does not rely on computational assumptions. We summarize this comparison in Table \ref{tab:comparison}.

Since we posted a first draft of our manuscript online, Aaronson and Hung~\cite{DBLP:conf/stoc/AaronsonH23} have made valuable contributions to the analysis of quantum advantage based certified randomness protocols. As noted in their work, our results have much in common. 
The main focus of~\cite{DBLP:conf/stoc/AaronsonH23} is to formalize the white-box certified randomness protocol introduced in~\cite{a-talk, a-talk0}, assuming the hardness of $\llqsv$ for $\mathsf{QCAM/poly}$, the basic outline of which preceded and inspired our work.
Their work also improves some of the results in this paper, in the sense that they give black-box evidence for hardness of $\llqsv$ and the security of the certified randomness protocol in the random oracle model, a line of inquiry which was instigated by our work.
We believe that both works together build a foundation to establish the first practical use-case of near-term quantum devices\footnote{Instantiating oracles using pseudorandom functions can have several challenges on NISQ devices that we leave open for future works.}. There are several differences between our works as well:
\begin{itemize}
\item Our result gives a $\mathsf{AM}$ (implied by our $\PH$ bound) and $\BQP$ lower bound for $\llqsv$ in the black box model, which Aaronson and Hung improve to $\mathsf{QCAM/poly}$. We note that our $\PH$ lower bound is incomparable to their result. Our work also finds connections to the Forrelation problem and uses different techniques that may be of independent interest.
\item Our techniques for studying Fourier Sampling in QROM allow us to argue about linear min-entropy directly without an entropy accumulation theorem. We utilize a ``derandomization'' argument to establish a general trade-off between a strong lower bound with constant min-entropy generation, and a weaker lower bound with linear min-entropy generation.
\item On the other hand, the proof techniques in~\cite{DBLP:conf/stoc/AaronsonH23}, work against more general \emph{entangled} quantum adversaries, where it is less clear how to extend our techniques in this way. We leave this extension as an interesting future direction.
\end{itemize}
\begin{table}[t]
    \centering
    \begin{tabular}{| l | l | l | l |}
        \hline
         \cellcolor[gray]{0.8} Protocol & \cellcolor[gray]{0.8}\cite{YK22} & \cellcolor[gray]{0.8} \cite{a-talk0} & \cellcolor[gray]{0.8} This Work, \cite{DBLP:conf/stoc/AaronsonH23} \\
         \hline
         \cellcolor[gray]{0.8} Verification & Efficient & Inefficient  & Inefficient \\
         \hline
         \cellcolor[gray]{0.8} Security & QROM + AA conjecture & Custom conjecture & QROM \\
         \hline
         \cellcolor[gray]{0.8} Implementation & Not NISQ  & RCS & RCS (heuristic) \\
         \hline
    \end{tabular}
    \caption{Comparison of \Cref{th:1} to the certified randomness results of \cite{YK22,a-talk0}. See \Cref{sec:rcs-connections} for a discussion of the connections between our construction and RCS.}
    \label{tab:comparison}
\end{table}

\subsection{Connecting Fourier Sampling to Aaronson's protocol}
\label{sec:rcs-connections}
The intellectual history of RCS as a quantum advantage experiment has deep roots in Fourier Sampling.
In 2010 Aaronson used a relation version of the Fourier Sampling problem called Fourier Fishing to provide the first oracle separation between the relation variants of  \textsf{BQP} and \textsf{PH} \cite{aaronsonbqpph}. 
The idea was to take advantage of the fact that the Fourier spectra of random Boolean functions obeys a ``Porter-Thomas'', or exponential, distribution in which some outcomes are sampled with considerably heavier probabilities than others.
Aaronson showed in the black box model that no PH algorithm can sample from mostly heavy Fourier outcomes, whereas a quantum computer naturally does so by performing Fourier Sampling.

Subsequently, it was conjectured that no classical algorithm can output mostly heavy elements of the output distribution of random circuits \cite{boixo-characterizing,ac,ag}.
These works essentially conjecture that the output distribution of random quantum circuits is as ``unstructured'' as the Fourier spectrum of a random Boolean function. In other words, the classical description of the circuit does not leak information about which outputs are heavy vs light in any efficiently accessible manner.
In contrast, there are certain scenarios in which such a conjecture fails.
For example, outputting heavy elements of Boson Sampling experiments is computationally easy due to information leakage from the circuit description \cite{aaronson2013bosonsampling}.
More recently, it has been shown that sufficiently shallow RCS circuits have additional structure that allows for classical algorithms to output heavy items with nontrivial bias \cite{gao2021limitations}.
Therefore in this sense, the classical hardness of quantum advantage experiments rests on this ``unstructuredness'' conjecture of deep RCS circuits.

Viewed in this light, our results can be seen as the analog of Aaronson's Fourier Fishing lower bound \cite{aaronsonbqpph} in the certified randomness setting.
We are the first to demonstrate that certified randomness is possible in the black-box Fourier Sampling context, as we show that no quantum algorithm can output heavy Fourier elements with low entropy.
Aaronson's prior work is conjecturing that precisely this same property holds in the white-box RCS setting.
The security of his protocol can therefore be rested on the belief that RCS has highly ``unstructured'' output distributions -- even hidden to (pseudo-deterministic) quantum algorithms.
In this sense, RCS can be viewed as a heuristic implementation of our protocol.

\subsection{Our Results} \label{subsec:our-results}

Our first result concerns the problem of outputting sufficiently heavy Fourier coefficients, which we call Fourier HOG ($\mathsf{FHOG}$). Let $F = f_1, \ldots, f_m: \{0, 1\}^n\rightarrow \{\pm 1\}$ be a list of random Boolean functions, and $Z = z_1, \ldots, z_m \in \{0, 1\}^n$. Let $N = 2^n$. We call a Fourier coefficient heavy if $\frac{1}{N} < \hat{f}_i(z_i)^2 \leq \frac{4}{N}$, and light if $\hat{f}_i(z_i)^2 \leq \frac{1}{N}$.  Intuitively, $\mathsf{FHOG}$ captures the statistical properties of the simple quantum algorithm that samples according to the Fourier spectrum.
A simple calculation shows that on a random $f$ this simple quantum algorithm outputs a heavy element $\approx 54\%$ of the time, and outputs a light element $\approx 20\%$ of the time. 
Therefore, $\mathsf{FHOG}$ requires that the number of heavy elements among $z_1, \ldots, z_m$ is $\approx 0.54 m$ and the number of light elements is $\approx 0.2 m$.
Our main result proves that \emph{any} efficient quantum query algorithm that solves $\mathsf{FHOG}$ on a non-negligible fraction of inputs must generate linear min-entropy, and hence the simple Fourier sampling algorithm is essentially optimal in this sense. Starting from a \emph{random} function to generate randomness might seem circular; however, this bound is on min-entropy \emph{conditioned} on the choice of random function.

\begin{theorem}{(informal)} \label{th:1}
There is no sub-exponential-query quantum algorithm which, given black-box access to $F$, both (i) passes $\mathsf{FHOG}$ with non-negligible probability, and (ii) conditioned on passing $\mathsf{FHOG}$, generates only $o(n)$ bits of min-entropy, for a non-negligible fraction of random functions~$F$.
\end{theorem}
This result can directly be used to construct a \emph{single round} publicly verifiable proof of min-entropy protocol:
\begin{definition}
    Let $n, m \in \mathbb{N}$ and $m \geq n$. Let $F = f_1, \ldots, f_m$ be a random oracle representing $m$ random boolean functions $f_i: \{0, 1\}^n\rightarrow\{\pm 1\}$. Suppose $A^F$ is a polynomial time quantum algorithm that generates polynomial size classical strings $\pi$. Let $V^F$ be a randomized exponential time classical verifier that accepts or rejects $\pi$. $A$ and $V$ form a single round publicly verifiable proof of min-entropy protocol if and only if:
    \begin{itemize}
        \item Completeness: $V$ accepts string $\pi$ sampled according to the output distribution of $A$ with high probability:
        $$\Pr_{F, \pi \sim A^F} \left[V^F(\pi) \text{ accepts}\right] \geq 1-\negl(n).$$
        \item Soundness: With high probability over the choice of random oracle $F$, if $V$ accepts $\pi$ with non-negligible probability, then the output distribution of $A^F$ must have min-entropy at least $h$:
        $$\Pr_{F} \left[\Pr_{\pi \sim A^F}[V^F(\pi) \text{ accepts}] \geq \delta(1/n) \land \Pr_{Z \sim \chi}[q_Z \geq 2^{-h}] \geq \negl(n) \right] \leq \negl(n).$$
        where $\chi$ is the output distribution of $A^F$ conditioned on $V$ accepting, and $q_Z$ is the probability of sampling $Z$ according to $\chi$
    \end{itemize}
\end{definition}
Which immediately implies a single round publicly verifiable certified randomness protocol \cite{YK22}.
\begin{remark}
Certified randomness usually refers to \emph{private} randomness expansion, where the goal is to expand a truly small random seed to a larger string that is statistically close to uniformly random by interacting with a quantum device. This definition is widely used in protocols based on Bell experiments \cite{Pironio2010}, where the only assumption they make is that they have access to multiple devices that cannot communicate. Informally, they can show that a strategy that wins the CHSH game with high enough probability must also have high min-entropy, which can be used to extract uniformly random bits using randomness extractors. Furthermore, they can even show that these protocols are secure against quantum adversaries. In other words, even when the devices are manufactured adversarially, the conditional min-entropy (conditioned on the adversary's sub-system) must still be large. Motivated by Aaronson’s proposal \cite{a-talk}, in this work we are concerned with \emph{publicly} verifiable certified randomness relative to quantum random oracles, since the (adversarial) prover knows the string generated in the protocol. This makes our problem slightly easier to argue about since we do not need to worry about the adversary’s side information. However, similar to \cite{YK22}, we showed that the high min-entropy bound holds even conditioned on the random oracle.
\end{remark}
Crucially, our bounds hold not only for worst-case functions $F$ but also for \emph{average-case} $F$, whose Fourier spectra share many common features with RCS (e.g. a Porter-Thomas-like  distribution). We prove this theorem using a carefully constructed hybrid argument.
The worst-case, constant-min-entropy proof is relatively straightforward: Suppose one has an efficient quantum algorithm that pseudo-deterministically outputs a heavy item for all functions $f$. 
On one hand, we know that by changing a relatively small number of entries of $f$ (say $O(2^{n/2})$), we can change a heavy Fourier coefficient to a light Fourier coefficient.
On the other hand, any low-query quantum algorithm cannot notice the difference of switching these many entries. Simply because they represent an exponentially small fraction of the entries of $f$, and therefore this case follows from a standard hybrid argument.
However, turning this idea into one that certifies \emph{linear} min-entropy for \emph{average-case} functions is much more involved. This is essentially because the above perturbation of the function $f$ does not preserve the average case, nor does it apply to quantum algorithms with higher min-entropy. 
We tune the hybrid argument to extend our bound to the average case (for constant min-entropy) by applying a combinatorial accounting of how these perturbations of the functions affect the average case. We then apply a novel ``derandomization'' argument which reduces the linear min-entropy case to the constant min-entropy case.

We then show how to leverage this result to produce a cryptographic certified randomness protocol, which is analogous to Aaronson's protocol.
The basic idea is to ask the quantum algorithm to sample from the Fourier spectra of several different random functions and check that the number of sufficiently ``heavy'' and ``light'' outcomes is roughly what is expected.
This could be instantiated in a white-box setting as a cryptographic protocol for certified randomness, where the client sends several function descriptions (say of cryptographic hash functions), the honest server \emph{efficiently} responds with Fourier spectrum samples, and the client checks (post-mortem) that most of the samples were heavy using an exponential time computation. We emphasize that this is only plausible because we can bound the \emph{conditional} min-entropy (relative to the choice of random function).

Our second result concerns a certified randomness proposal due to Aaronson~\cite{a-talk}. Aaronson gave a reduction from the problem of certifying min-entropy to a certain decision problem called the Long List Quantum Supremacy Verification ($\mathsf{LLQSV}$) problem.
The $\mathsf{LLQSV}$ problem essentially evaluates the difficulty of distinguishing exponentially many RCS samples from uniformly random numbers.
Aaronson showed that if this problem lies outside of $\mathsf{QCAM}$ -- that is the class of problems solvable by $\mathsf{AM} $ protocols with classical messages and quantum verifiers -- then no efficient quantum algorithm can pass the HOG test with low min-entropy.
Analogously to \cite{ac}, Aaronson explored potential  algorithms for this problem and found they did not solve it, and subsequently conjectured such a $\mathsf{QCAM}$ algorithm does not exist.
This is a custom assumption, and therefore deserves scrutiny -- it is simultaneously conjecturing this problem lies outside of $\BQP$, outside of $\mathsf{AM}$, and beyond. 
In our second result, we give evidence towards its validity in the black-box setting.
Namely, we show that the black-box $\mathsf{LLQSV}$ problem lies outside of $\mathsf{BQP}$ and even outside of $\mathsf{PH}$ (and hence outside of $\mathsf{AM}$).  

\begin{theorem}{(informal)} \label{th:2}
No $\mathsf{BQP}$ or $\mathsf{PH}$ algorithm can solve black-box $\mathsf{LLQSV}$.  Furthermore, our lower-bound extends even to the case in which the quantum algorithm is allowed to make $O(2^{n/10})$ queries, and the relation in the $\mathsf{PH}$ algorithm can be computed in time $O(2^{n/c})$ for some constant $c$ that depends on the level of the polynomial hierarchy.
\end{theorem}

We give two independent proofs of the $\mathsf{PH}$ lower bound and one of the $\mathsf{BQP}$ lower bound. The first $\mathsf{PH}$ bound and the $\mathsf{BQP}$ bound are based on extending lower bounds for the MAJORITY function.
The second $\mathsf{PH}$ bound is based on introducing a black-box variant of $\llqsv$ that we call $\sqfor$, which is also closely related to the $\mathsf{Forrelation}$ problem.
In the problem we introduce, black-box access is given to two functions $f, g: \{0, 1\}^n \rightarrow \{\pm 1\}$, and the task at hand is to distinguish whether $f, g$ are chosen uniformly at random or $g$ is correlated with the heavy Fourier coefficients of $f$.
It is simple to show that the $\sqfor$ problem is strictly easier than $\llqsv$, because it reduces to $\llqsv$ and it also admits a $\mathsf{BQP}$ algorithm. Intuitively, in $\sqfor$, having access to $g$ allows one to sample heavy elements of the Fourier spectrum and therefore recreate a long list of samples for $\llqsv$. Nevertheless, we prove that even the $\sqfor$ problem is not in $\PH$. This simultaneously proves a lower bound on black-box $\llqsv$ and gives another oracle separation of $\mathsf{BQP}$ vs $\mathsf{PH}$~\cite{DBLP:conf/stoc/RazT19}.  

\subsection{Open Problems}
We have shown that Fourier Sampling gives rise to certified randomness in the QROM, albeit without efficient verification. In contrast, Yamakawa-Zhandry protocol~\cite{YK22} gives certified randomness in the QROM with efficient verification, but relies on the Aaronson-Ambainis conjecture. 
It is natural to ask if one can combine the strengths of both results and obtain a certified randomness protocol with efficient verification and unconditional security.

Another natural direction is to construct certified randomness protocols without the QROM heuristic in the white-box model. Aaronson's protocol~\cite{a-talk} is such an attempt that relies on the $\llqsv$ conjecture, or alternatively, a belief that the output distribution of sufficiently deep random circuits is unstructured. It is an important open question whether or not one can base the security of these protocols on more natural hardness assumptions.

\subsection{Proof Outline}
\subsubsection{Min-Entropy Generation}
\label{subsec:th1-outline}
Our proof of \Cref{th:1} proceeds in three steps, where each step is a progressively stronger theorem.

\paragraph{Worst-case, constant min-entropy bound.} In the first step, we show that any quantum algorithm that outputs a heavy Fourier coefficient for \emph{any} function $f$ after $2^{o(n)}$ queries must generate a {\it constant} amount of min-entropy. A simple hybrid argument suffices for this case as sketched in \cref{subsec:our-results}.

\paragraph{Average-case, constant min-entropy bound.} Unfortunately, the above argument does not immediately say anything about quantum algorithms which only output a heavy Fourier coefficient for \emph{most} functions $F$. This is essentially because the hybrid argument only shows how to \emph{adversarially} construct ``Bad'' functions for which the quantum algorithm does not solve $\mathsf{FHOG}$ with non-negligible probability. It does not show that these Bad functions constitute a large fraction of the space. 

In the second step of our proof, we prove exactly this -- any low-query algorithm cannot pseudo-deterministically solve $\mathsf{FHOG}$ on a large fraction of inputs. In fact, we can show that the algorithm succeeds only on exponentially small choices of $F$.
We prove this using a double-counting argument.
Namely, given a quantum algorithm $A$ that solves $\mathsf{FHOG}$ with low min-entropy, label the functions on which the algorithm succeeds ``Good'', and those on which it \emph{noticeably} fails ``Bad.'' More formally, for a Bad $F$, approximately $28\%$ of outcomes are ``heavy'' and $46\%$ of outcomes are ``light''.
For each Good choice of $F$ on which the algorithm passes $\mathsf{FHOG}$, the hybrid argument constructs a nearby Bad function $F'$.
Consider the bipartite graph where one side contains Good functions $F$, the other side contains Bad functions $F'$, and we draw an edge $(F, F')$ if $F'$ can be obtained from $F$ by our earlier construction. The hybrid argument gives us a lower bound on the degree of any Good vertex in this graph. However, to show the Bad choices are a large fraction of the space, we instead want to upper bound the degree of any Bad vertex. 
We obtain such a bound by counting the number of adjacent functions $F$ for which the most likely output of our algorithm on $F'$ is a heavy Fourier coefficient.
In other words, we prove that the set of Bad functions is actually very large by showing that each Bad function $F'$ could only have come from a limited number of Good functions $F$. 
Therefore the Bad functions must constitute a large fraction of the space -- which concludes our average-case bound.

\paragraph{Average-case, linear min-entropy bound.} For the previous step, we crucially relied on the fact that we were only considering quantum algorithms with sub-constant min-entropy so that there was a unique majority output $Z$. This allowed us to consider shifting entries of $F$ to change $Z$ from a heavy to a light element and vice versa.
Once we are in the case of algorithms with higher min-entropy $H$, this type of argument becomes much more difficult since there could now be many different elements $Z_1,Z_2,\ldots Z_{2^H}$ which are frequently output by the algorithm.
The simplest approach to proving a hybrid bound is to try to shift around the heaviness of the Fourier coefficients of various subsets of these elements without affecting the others. But this becomes challenging, as the number of entries of $F$ one can change to alter certain Fourier coefficients but not others decays exponentially in the number of coefficients being considered. So our attempt to make a direct extension of our prior method fails in the average case.

Instead, we extend our theorem to the generation of a linear amount of min-entropy with a derandomization-like argument that uses the previous theorem as a black-box. 
While it's commonly said that one cannot ``pull the randomness'' out of quantum algorithms, in certain query settings it can be done.
In particular, whenever one has a quantum algorithm $A$ which has an output distribution of min-entropy $H$, there is a simple way to pull the quantum randomness out of the algorithm: Run the quantum algorithm many ($\approx 4^H$) times, store the empirical distribution, and then classically rejection sample on the ``heavy'' portion of the empirical distribution.
The empirical distribution will be close to the true distribution on heavy elements, so this new sampler $A'$ will faithfully reproduce the distribution of heavy outputs for a sufficiently large fraction of input choices.
Critically, the randomness in $A'$ is now almost entirely classical, stemming from rejection sampling. This randomness can be ``hard-wired'' externally (modulo some technical issues regarding the imprecision caused by the additive error in the estimated distribution from the samples). 
For any fixed external randomness, the algorithm becomes pseudo-deterministic, at the cost of more queries.

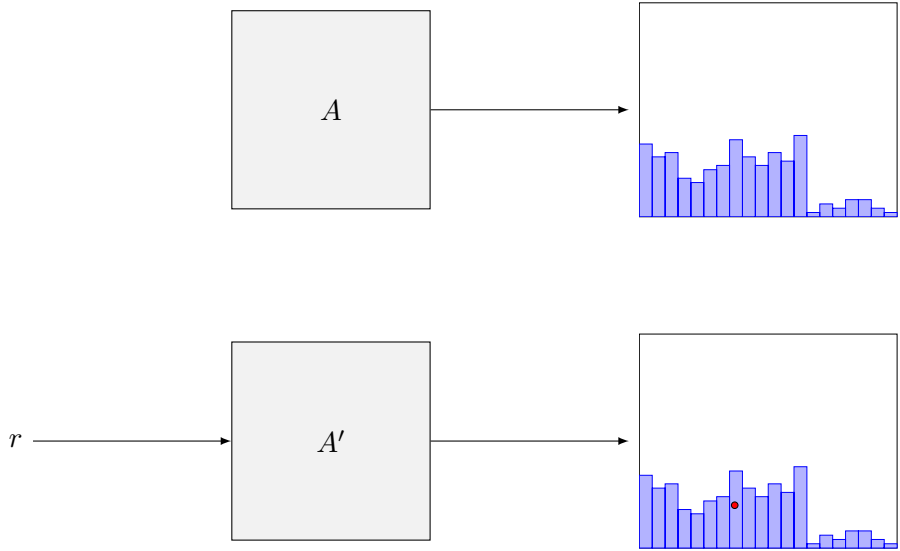
\begin{figure}
    \centering
    \begin{tikzpicture}[>=latex,rect/.style={draw=black, 
                       rectangle, 
                       fill=gray,
                       fill opacity = 0.1,
                       text opacity=1,
                       minimum width=75pt, 
                       minimum height = 75pt, 
                       align=center}]
      \node[rect] (a1) {$A$};
      \node[right=75pt of a1] (a2) {\includegraphics[scale=0.5]{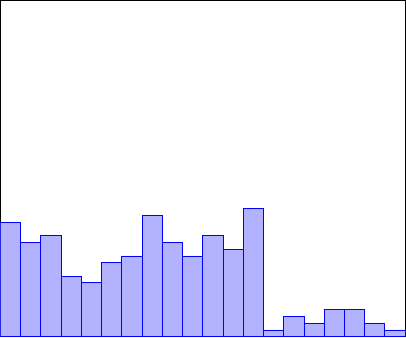}};
      \node[rect,below=50pt of a1] (b1) {$A'$};
      \node[right=75pt of b1] (b2) {\includegraphics[scale=0.5]{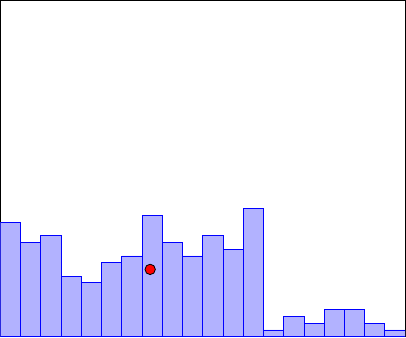}};
      \node[left=75pt of b1] (b0) {$r$};
      \draw[->] (a1)--(a2)node[midway,sloped, above] {};
      \draw[->] (b1)--(b2)node[midway,sloped, above] {};
      \draw[->] (b0)--(b1)node[midway,sloped, above] {};
    \end{tikzpicture}
    \caption{The algorithm $A'$ takes enough samples from $A$ to estimate the output distribution, then uses $r$ as a seed to rejection sample from that distribution. For random $r$, the output distributions of $A$ and $A'$ are identical. But for fixed $r$, the output of $A'$ is nearly-deterministic.}
    \label{fig:pulling-out-the-randomness}
\end{figure}

At first glance, the argument might appear overly expensive because it significantly increases the number of queries. Specifically, if $A$ makes $Q$ queries, then $A'$ makes $\approx Q \cdot 4^H$ queries, which in our case is exponential. 
Surprisingly, we show that it is still enough to reduce the constant min-entropy case to the linear min-entropy case.
The crucial observation is that, for the constant min-entropy case, our query lower bound holds against \emph{exponentially many} queries. Therefore, so long as $Q \cdot 4^H$ is less than this exponential, $A'$ does not make too many queries to $F$. We then show that one can hard-code the randomness for rejection sampling into our new algorithm, making it pseudo-deterministic and therefore ``derandomizing'' linear min-entropy algorithms.
In summary, we have traded our constant-min entropy, large-exponential query bound for a linear min-entropy, smaller-exponential query bound. 

\subsubsection{The \texorpdfstring{$\sqfor$}{Squared Forrelation} Problem and \texorpdfstring{$\bllqsv$}{Black box llqsv}} \label{subsec:intro-llqsv}

In this section, we discuss $\mathsf{PH}$ and $\mathsf{BQP}$ lower bounds for the $\bllqsv$ problem. First, we formally define the problem:
\begin{problem} [$\bllqsv$]
Given oracle access to $\mathcal{O}$, distinguish the following two cases:
 \begin{itemize}
  \item $\mathcal{U}$: For each $i$, $f_i$ and $s_i$ are chosen uniformly at random.
  \item $\mathcal{D}:$ For each $i$, $f_i$ is a random Boolean function, and $s_i$ is sampled according to $\hat{f}_i^2$.
 \end{itemize}
\end{problem}
 
We observe there is a close connection between $\bllqsv$ and the $\mathsf{MAJORITY}$ problem. To understand this, note that in either case the probability of sampling a pair $(f, s)$ only depends on the magnitude of the Fourier coefficient of $f$ at $s$, or equivalently $|h(f\cdot\chi_s)-N/2|$, where $h$ is the hamming weight, $\chi_s$ is the $s$-Fourier character (i.e., $\chi_s(x) = (-1)^{x\cdot s}$), and $f\cdot \chi_s$ is the entry-wise product of truth tables of $f$ and $\chi_s$. 
This allows us to write both distributions as a linear combination of simpler distributions $\mathcal{U}_d^N$ -- the uniform distribution over all strings of hamming weight $N/2+d$ or $N/2-d$. 

An averaging argument then implies that any algorithm for $\bllqsv$ should also distinguish a pair of distributions: $\mathcal{U}_d^N$ and $\mathcal{U}_{d'}^N$.  To prove the hardness of this task, we can focus on the distinguishability of $\mathcal{U}_0^N$ and $\mathcal{U}_d^N$, and since we are working with random Boolean functions, we can bound $d$ by $\poly(n)/\sqrt{N}$. Lastly, we use known lower bounds for $\mathsf{MAJORITY}$ to prove that these two distributions are indistinguishable for both $\mathsf{BQP}$ and $\mathsf{PH}$ algorithms. We discuss this in more detail in Section \ref{sec:maj}.

Perhaps more interestingly, we can show a reduction from $\sqfor$, to $\bllqsv$.
This enables us to get an alternative lower bound for $\bllqsv$, by proving a $\mathsf{PH}$ lower bound for $\sqfor$. Crucially, it is not hard to see that $\sqfor$ is in $\mathsf{BQP}$, which makes it a strictly easier problem than $\bllqsv$, and hence our lower bound is stronger. $\sqfor$ also gives an alternative oracle separation of $\mathsf{BQP}$ and $\mathsf{PH}$, that has distinct advantages from previous similar results.

One advantage of $\sqfor$ is that it can be solved by a quantum algorithm that only makes {\em classical} queries to $g$. This property plays a critical role in the reduction between $\sqfor$ and $\bllqsv$.
Furthermore, this property provides another benefit. It enables us to consider a more realistic instantiation of these oracles, especially when compared to \cite{DBLP:conf/stoc/RazT19}, as we discuss next.

Recall that in \cite{DBLP:conf/stoc/RazT19}, an oracle separation $\mathsf{BQP}^O \not\subseteq \PH^O$ was demonstrated using the Forrelation problem. Given oracle access to two functions $f$ and $g$, the problem asks to distinguish between two scenarios: $(i)$ $f$ and $g$ are completely random and independent or $(ii)$ the Fourier transform of $f$ correlates with $g$. The Forrelation problem can be solved via a simple efficient quantum algorithm that queries both $f$ and $g$ in superposition, but cannot be solved by a $\PH$ machine. One drawback of this problem is that there's seemingly no efficient implementation of the oracles under which the problem remains non-trivial. In particular, if $f$ is a random  small circuit, then it seems any small circuit $g$ doesn't correlate with $f$'s Fourier transform. If one settles for having small circuits only for $f$ but not $g$ (or vice versa), then the problem remains non-trivial, but even in this scenario, the problem seems far from practical as a quantum algorithm would need to query an exponentially large oracle $g$ in superposition. $\sqfor$ is defined as follows:
 \begin{problem} [simplified $\sqfor$] \label{prob:simplifiedsqfor}
   Given oracle access to Boolean functions $f,g: \{0, 1\}^n \to \{\pm 1\}$, distinguish the following two cases:
\begin{enumerate}
  \item $f$ and $g$ are both chosen uniformly at random.
  \item $f$ is chosen uniformly at random and $g$ is the indicator of the heavy Fourier coefficients of $f$.
\end{enumerate}
\end{problem}
A Fourier coefficient is heavy if the coefficient squared is larger than its expected value $\frac{1}{N}$.
The main difference between $\sqfor$ and $\mathsf{Forrelation}$ is that $g$ is correlated with \emph{squared} Fourier coefficients rather than the sign of Fourier coefficients. Clearly, a simple quantum algorithm that queries $f$ in superposition, measures in the Hadamard basis, and then queries $g$ classically solves the problem.
Notice that this algorithm is a combination of a quantum and classical algorithm. This means that in principle, if $f$ has a small circuit, then the quantum part of the algorithm can implemented efficiently, followed by a much longer classical ``post-mortem'' part that queries $g$. It is important to note that in both the $\mathsf{Forrelation}$ and $\sqfor$ problems, $g$ generally is not efficiently computable, even when $f$ can be succinctly represented.

In particular, following~\cite[Section~7.4]{ac}, this suggests choosing the oracle $f$, the only function that is queried in superposition, from a pseudorandom family of functions. While this manuscript does not delve into these aspects, it is worth noting a distinction: our problem is a decision problem, contrasting with the problems presented in~\cite[Section~7.4]{ac} -- known as Fourier Fishing and Fourier Sampling -- which are sampling problems.

\paragraph{Lower Bound for $\bllqsv$ via $\sqfor$}
 We summarize the second $\mathsf{PH}$ lower bound for $\bllqsv$ in two steps. First, we prove a $\mathsf{PH}$ lower bound for $\sqfor$, and then we show a reduction from $\sqfor$ to $\bllqsv$. 

The first step of the proof proceeds by carefully incorporating the \emph{squared} Fourier coefficients in the analysis of Raz and Tal~\cite{DBLP:conf/stoc/RazT19}.
Recall that the analysis of \cite{DBLP:conf/stoc/RazT19} proceeds by looking at a multivariate Gaussian as a stopped Brownian motion. This allows them to write the expected value of a multilinear function over the multivariate Gaussian as a telescopic sum, and bound each term in the sum independently. 
Finally, they show that it is possible to convert this continuous distribution to a discrete distribution $\mathcal{D}$ while preserving the expected value by truncating each coordinate to the interval $[-1,1]$ followed by randomized rounding.

A feature of Gaussians, which was heavily exploited in the analysis of \cite{DBLP:conf/stoc/RazT19}, is that the sum of $t$ independent Gaussians is again a Gaussian.
When dealing with squares of Gaussians, this is no longer true. Nevertheless, if $Z\sim \mathcal{N}(0,\sigma^2)$ then one can still write $Z^2 - \E[Z^2]$ as the result of a martingale composed of many small steps. To present $Z^2 - \E[Z^2]$ as a martingale, first observe that if $X$ (the ``past'') and $\Delta$ (the ``next step'') are two independent zero-mean Gaussians, then $$(X+ \Delta)^2 - \E[(X+\Delta)^2] = (X^2 - \E[X^2])   + \Delta^2 + 2\Delta \cdot X - \E[\Delta^2]$$
and thus $\E[\Delta^2 + 2\Delta \cdot X - \E[\Delta^2]\;|\;X] = 0$. Since a Gaussian $Z\sim \mathcal{N}(0,\sigma^2)$ can be written as a sum of $t$ independent Gaussians $Z_1, \ldots, Z_t\sim \mathcal{N}(0,\sigma^2/t)$,  we see that $Z^2-\E[Z^2]$ may be written as sum of $t$ small steps, each of the form $Z_i^2 + 2Z_i \cdot Z^{< i} - \E[Z_i^2]$, with expected value $0$ even conditioned on the past $Z^{<i} = Z_1 + \ldots +Z_{i-1}$.
The rest follows from the analysis of \cite{DBLP:conf/stoc/RazT19}, which crucially relies on representing the final distribution as a sum of many small steps, each with an expected value of $0$
As a consequence, we obtain a new oracle $O$ so that 
$\mathsf{BQP}^O\not\subseteq\mathsf{PH}^O$. 

In our next step, the reduction from $\sqfor$ to black-box $\mathsf{LLQSV}$, we make use of the fact that $g$ is not queried in superposition by the quantum algorithm. This observation follows by rejection sampling according to $g$; intuitively, this allows one to sample heavy elements of the Fourier spectrum and therefore recreate a long list of samples. So far, we have worked with one instance of the $\sqfor$ problem. However, we prove that our lower bound easily extends to the case that we have oracle access to an (exponentially) long list of samples in $\sqfor$. Consequently, this shows that the $\bllqsv$ problem is not in $\mathsf{PH}$.

In our previous discussion, we made two oversimplifications which obscure the technical challenges in our argument.
First, unlike in our simplified $\sqfor$ problem, the actual $\sqfor$ problem concerns a distribution over functions $f$ and $g$ that are obtained by taking multivariate Gaussians and randomized rounding them to have support over the Boolean cube.  Second, rejection sampling according to $g$ samples a uniformly random heavy squared Fourier coefficient of $f$, which differs from Fourier transform distribution itself; we address these differences in detail in Section \ref{subsec:llf}. Intuitively, to fix this issue, recall that the goal of these long list hardness results is to generate certifiable random numbers. We show that this task is achievable even when the quantum algorithm can sample heavy outcomes from the distribution generated by rejection sampling (i.e., we can think of this as a new benchmark test much like $\mathsf{HOG}$, but with respect to a new distribution). Accordingly, we show that the same quantum algorithm that samples according to the Fourier transform distribution also outputs heavy outcomes from the rejection sampling distribution on average. To prove this, we need to combine the analysis of quantum algorithm's success probability for $\sqfor$ with tail bounds on the number of ``heavy'' coefficients.

\section{Low-Entropy sampling bounds against \BQP}

In this section, our main goal is to prove that $\mathsf{FHOG}$ can be used as proof of min-entropy in the Fourier sampling setting, and the trivial Fourier sampling algorithm is almost optimal. Since we are checking a statistical property of the output distribution for our test, and the prover can be adversarial, we need to directly analyze the multiple sample version of the problem to deal with parallel repetition, and avoid making assumptions about the algorithm being memoryless. Let $F = f_1, \ldots, f_m$ be a list of Boolean functions $f_i: \{0, 1\}^n \rightarrow \{\pm 1\}$ and $\hat{f_i}(z) = \E_{x\in\{0,1\}^n}[ f_i(x) \cdot (-1)^{z \cdot x}]$. We define Fourier HOG as follows:
\begin{problem}[$\mathsf{FHOG}$] \label{prob:HOG}
    Given oracle access to $F$, output a list of outcomes $Z = z_1, \ldots, z_m$, $z_i \in \{0, 1\}^n$ such that:
    $$ |\mathcal{H}(F,Z) - T_\mathcal{H}\cdot m| \leq c\cdot m $$  
    $$  |\mathcal{L}(F,Z) - T_\mathcal{L} \cdot m| \leq c \cdot m $$  
\end{problem}
\noindent where $\mathcal{H}$ and $\mathcal{L}$ are defined by: 
$$\mathcal{H}(F,Z) = \left\{i \in [m] \;\middle|\; \tfrac{1}{N} <  \hat{f}_i(z_i)^2 \leq \tfrac{4}{N} \right\},
\quad  
\mathcal{L}(F,Z) = \left\{i \in [m] \;\middle|\; \hat{f}_i(z_i)^2 \leq \frac{1}{N} \right\},$$
and 
$T_\mathcal{H} = \E_{Y\sim \mathcal{N}(0,1)}[Y^2 \cdot \mathds{1}_{\{1<Y^2\le 4\}}] \approx 0.54$,  
$T_\mathcal{L} =  \E_{Y\sim \mathcal{N}(0,1)}[Y^2 \cdot \mathds{1}_{\{Y^2\le 1\}}] \approx 0.2$ 
and $c=0.02$. As we show in section \ref{subsec:he-algorithm}, these constants ($T_\mathcal{H}$, $T_\mathcal{L}$ and $c$) are chosen in a way that a Chernoff bound can show that the simple quantum algorithm that samples according to the Fourier spectrum generates a valid solution for $\mathsf{FHOG}$ with high probability over the choice of functions and outcomes. 
\subsection{Worst-case, constant min-entropy: A simple BBBV argument}
Let
$Z_F = z_{f_1} \ldots z_{f_m}$ ($z_{f_i} \in \{0, 1\}^n$) be the most-probable output of $A^F$ (say, the first in lexicographical order if there are multiple such outputs).
Then, let
$$P_{f_i} = \{x \in \{0,1\}^n : f_i(x) \cdot (-1)^{z_{f_i} \cdot x} = \sgn(\hat{f_i}(z_{f_i}))\}.$$ In words, $P_{f_i}$ contains the inputs $x$ that contribute to $\hat{f}_i(z_{f_i})^2$ being large; note that $P_{f_i}$ must have at least $N/2$ items. Second, let $E(f_i)$ be the set of functions that differ from $f_i$ in exactly $\sqrt{N}/2$ places that are all contained in $P_{f_i}$. Let $Q \subseteq [m]$, we can also define 
$$E_Q(F) = \{f'_1 \ldots f'_m \in \{\pm 1\}^N \mid \forall i \in Q: f'_i \in E(f_i) \land \forall i \not\in Q: f'_i = f_i\}$$
$E_Q(F)$ is the set of all list of functions that differ with $F$ on elements of $Q$, and for each $i$ the new function is chosen from $E(f_i)$.
\begin{lemma}
\label{lemma:hybrid}
Let $A^F$ be a quantum algorithm that makes $T$ queries to $F$. For a $1-N^{-1/4}$ fraction of $F' \in E_Q(F)$, we have $\delta(A^F, A^{F'}) \le 2 \cdot |Q| \cdot T N^{-1/8}$ 
(where $\delta$ denotes the statistical distance between two distributions).
\end{lemma}
\begin{proof}
Our proof closely mirrors the hybrid argument of \cite{BBBV97}. At the $t$-th query to a phase oracle for $F$, let $\alpha_{t,x,w}^i$ be the amplitude on function $f_i$, query $x$, and workspace register $w$.
The ``query magnitude" on $\bigcup_i P_{f_i}$ is
$$\sum_{t=1}^T \sum_{i \in Q} \sum_{\substack{w\\x \in P_{f_i}}} |\alpha_{t,x,w}^i|^2 \le T$$
where $t$ is the time step and $w$ is the workspace register. Then for a random set $S = \bigcup_{i \in Q} S_i$, where $S_i \in \binom{P_{f_i}}{\sqrt{N}/2}$ defining a random function $f_i'(x) = f_i(x) \cdot (-1)^{1_{S_i}(x)} \in E(f_i)$,
$$\E\limits_S\left[\sum_{t=1}^T \sum_{i \in Q}\sum_{\substack{w\\x \in S_i}} |\alpha_{t,x,w}^i|^2\right] \le \frac{|Q|\cdot T \sqrt{N}}{2 \min_i |P_{f_i}|} \le \frac{|Q|\cdot T}{\sqrt{N}}.$$
By Markov's inequality,
$$\Pr_S\left[\sum_{t=1}^T \sum_{i \in Q}\sum_{\substack{w\\x \in S_i}} |\alpha_{t,x,w}^i|^2 \ge \frac{|Q|\cdot T}{N^{1/4}}\right] \le \frac{1}{N^{1/4}},$$
So for the remainder of the proof, we assume that
$$\sum_{t=1}^T \sum_{i \in Q} \sum_{\substack{w\\x \in S_i}} |\alpha_{t,x,w}^i|^2 \le \frac{|Q|\cdot T}{N^{1/4}}.$$
By Cauchy-Schwartz,
$$\sum_{t=1}^T \sqrt{\sum_{i \in Q}\sum_{\substack{w\\x \in S_i}} |\alpha_{t,x,w}^i|^2} \le \frac{|Q| \cdot T}{N^{1/8}}.$$
Then by a hybrid argument, we can bound the total variation distance as
\begin{align*}
    \delta(A^F, A^{F'}) \le ||A^F \ket{0} - A^{F'} \ket{0}||_2 \le \frac{2\cdot|Q|\cdot T}{N^{1/8}}.\tag*{\qedhere}
\end{align*}
\end{proof}
In particular, this lemma implies that $\Pr[A^{F'}=z_F] \ge \Pr[A^F=z_F] - 2 \cdot |Q|\cdot T N^{-1/8}$ for a $1-N^{-1/4}$ fraction of $F' \in E_Q(F)$. From this it is already easy to see that, for any algorithm $A$ with very-low min-entropy (so that a decent majority of the time $A^F$ outputs one specific $Z$), there exist functions $F'$ for which $A^{F'}$ will fail to pass $\mathsf{FHOG}$ on $F'$. This follows since given a pair $(F, Z)$ such that $Z$ passes $\mathsf{FHOG}$, we can simply set $Q$ to contain the index $i$ of all heavy outcomes $z_{f_i}$, and convert all of them randomly to light elements by picking a random function from $E(f_i)$. 
Next, we show that a similar min-entropy bound holds even when the algorithm only works on \emph{average} over the choice of functions.

\subsection{Average-case, constant min-entropy: Vertex degree in a bipartite graph}
\noindent To extend this to an average-case statement about a list of random Boolean functions, first we notice that a direct application of the hybrid argument fails. Even though the worst-case argument can construct a large set of ``Bad'' functions (in which $A^F$ is not a solution for $\mathsf{FHOG}$) from any list of ``Good'' functions (where $A^F$ passes $\mathsf{FHOG}$), the size of all ``Bad'' lists might still be negligible compared to the size of ``Good'' functions. The concern is that many ``Good'' functions get pushed to the same small set of Bad ones. However, we show that under a careful double-counting argument, we can effectively bound the number of ``Good'' inputs that can produce any specific ``Bad'' input.

Let $\mathsf{Good}$ be the set of inputs that $A$ has constant min-entropy ($h = -\log(2/3)$) on, and the outcomes pass $\mathsf{FHOG}$:
\begin{align*}
\mathsf{Good}= \left\{F=f_1\ldots f_m \;\middle|\;  \begin{aligned} & |\mathcal{H}(F, Z_F)- T_\mathcal{H} \cdot m | \leq  c\cdot m \\ 
&\land |\mathcal{L}(F, Z_F) - T_\mathcal{L} \cdot m| \leq c\cdot m  \\
&\land \Pr[Z_F \sim A^F] \geq 2^{-h}\end{aligned}\right\}
\end{align*}
We let $\mathsf{Bad}$ be the set of inputs that the most probable outcome is ``pretty bad'', so most outcomes are light rather than heavy. Very specifically, we set $T^b_\mathcal{H} = T_\mathcal{H} - 0.26 \approx 0.28$ and $T^b_\mathcal{L} = T_\mathcal{L} + 0.26 \approx 0.46$ and define 
\begin{align*}
\mathsf{Bad}= \left\{F=f_1\ldots f_m \;\middle|\;  \begin{aligned} & |\mathcal{H}(F, Z_F) - T^b_\mathcal{H}\cdot m| \leq c\cdot m \\ 
&\land |\mathcal{L}(F, Z_F)- T^b_\mathcal{L} \cdot m| \leq c\cdot m  \\
&\land \Pr[Z_F \sim A^F] \geq 2^{-(h+0.01)}\end{aligned}\right\}
\end{align*}
We now show that $|\mathsf{Bad}|$ is exponentially larger than $|\mathsf{Good}|$, which implies that conditioned on $A$ being almost deterministic, $A$ passes $\mathsf{FHOG}$ with negligible probability (over the choice of random functions).
\newcommand{\Good}{\mathsf{Good}}
\newcommand{\Bad}{\mathsf{Bad}}
\begin{theorem}
\label{thm:entropy}
Suppose that $A$ makes at most $O(N^{1/9})$ queries to $F$. Then,
$$|\Bad| \ge \exp(\Omega(m))\cdot |\Good|$$
\end{theorem}
\begin{corollary}
\label{corollary:entropy}
Suppose that $A$ makes $O(N^{1/9})$ queries to $F$. Then, $\Pr_F \left[F \in \Good\right] \leq \exp(-\Omega(m))$.
\end{corollary}
The corollary simply follows since $\Pr_F\left[F \in \Good\right] + \Pr_F\left[F \in \Bad\right] \leq 1$. Note that we can choose $m$ to be any large $\poly(n)$, and make this probability as small as we need in next steps.
\begin{proof}[Proof of \Cref{thm:entropy}]
Consider the following bipartite graph on nodes $\Good$ and $\Bad$. 
Assume (without loss of generality) that $m$ is multiple of $50$ and hence $0.26 m$ is an integer. 
Draw an edge between $F \in \Good$ and $F' \in \Bad$ if $\exists Q, |Q| = 0.26m: F' \in E_Q(F)$ and $\delta(A^F, A^{F'}) \le 2 \cdot |Q| \cdot T N^{-1/8}$. 
By \Cref{lemma:hybrid}, there are at least
$$(1-N^{-1/4}) \cdot \sum_{F\in \Good, Q} |E_Q(F)| \ge (1-N^{-1/4}) \cdot \binom{0.52m}{0.26m} \binom{N/2+\sqrt{N}/2}{\sqrt{N}/2}^{0.26 \cdot m} |\Good|$$
edges incident on $\Good$
 since for $(F,F')$ with $H_\infty(A^F) \le h$ and $\delta(A^F, A^{F'}) \le 2 \cdot |Q| \cdot T N^{-1/8} = o(1)$ we must have $H_{\infty}(A^{F'}) \le h + o(1) \le h + 0.01$.
Moreover, if $(F, F')$ is an edge in the graph, since $\Pr[A^F = z_F] = 2^{-H_\infty(A^F)} \ge 2^{-h}\ge 2/3$ and $\delta(A^F, A^{F'})=o(1)$ then $A^F$ and $A^{F'}$ have the same majority output (i.e., $Z_F = Z_{F'}$). This allows us to bound the number of edges on $\Bad$ from above:
$$\binom{0.48m}{0.26m}\binom{N/2}{\sqrt{N}/2}^{0.26 m} |\Bad|$$
edges incident on $\Bad$. Combining these bounds, we find that
\begin{align*}
    \frac{|\Bad|}{|\Good|} &\ge (1-N^{-1/4}) \cdot \frac{
    \binom{0.52m}{0.26m}
    \cdot \binom{N/2+\sqrt{N}/2}{\sqrt{N}/2}^{0.26 m} }
    {\binom{0.48m}{0.26m} \cdot \binom{N/2}{\sqrt{N}/2}^{0.26 m}
    } 
\ge (1-o(1)) \cdot \left(\frac{0.52}{0.48}\right)^{0.26 \cdot m} \cdot e^{\frac{1}{2}0.26 \cdot m}\qedhere
\end{align*}
\end{proof}
Intuitively, this theorem implies that the constant min-entropy portion of any quantum algorithm does not solve $\mathsf{FHOG}$ with high probability over the choice of functions. In other words, there is no quantum algorithm that has constant min-entropy on $\gg \exp(-m)$ fraction of inputs, if the heavy outcomes are good solutions to Problem \ref{prob:HOG}.  Next, we use this fact to extend this lower bound even to quantum algorithms with linear min-entropy by giving a reduction to the constant min-entropy version.

\subsection{Average-case, linear min-entropy: A general reparameterization trick}
Notice that we have an exponential query bound for a constant min-entropy algorithm. Using exponentially many queries, we can approximate the output distribution of a linear min-entropy algorithm. This allows us to trade our exponential query lower bound against \emph{constant} min-entropy for a smaller-exponential query lower bound against \emph{linear} min-entropy.

\begin{theorem}
\label{thm:linear-entropy}
Suppose $A$ is any quantum algorithm making $O(N^{1/100})$ queries to an oracle, and let $F = f_1, \ldots, f_m : \{0,1\}^n \to \{\pm1\}$ be a list of random Boolean functions, with $m \ge n$. For any polynomial $\delta(1/n)$, there exists a negligible function $\negl(n)$ such that the following holds:
$$\Pr_F \left[\Pr_A[A^F \text{ passes } \mathsf{FHOG}] \geq \delta(1/n) \land \Pr_{Z \sim \chi}[q_Z \geq 2^{-n/40}] \geq \negl(n) \right] \leq \negl(n) $$
where $\chi$ is the output distribution of $A^F$ conditioned on passing $\mathsf{FHOG}$ and $q_Z$ is the probability to sample $Z$ according to $\chi$.
\end{theorem}
In words, the theorem states that it is very unlikely to draw a collection of functions $F$ such that $A^F$ passes $\mathsf{FHOG}$ with non-negligible probability and the output of $A^F$ has high enough probability mass on elements that appear with probability at least $2^{-n/40}$. Note that it is crucial for us to analyze $\chi$ -- the output distribution of $A$ conditioned on $F$ -- and not the output distribution of $A$ averaged over random oracles. This implies that the output distribution of $A$ must be statistically close to a high min-entropy distribution even when conditioned on the choice of the random oracle, and hence, the generated entropy is inherently from the quantum algorithm.
\begin{proof}
Fix $\negl(n) = \gamma$, and suppose that there exists a quantum algorithm $A$ that does not satisfy the condition. This implies that on $1/\poly(n)$ fraction of inputs $A^F$ passes the test with probability at least $1/\poly(n)$ and that  $\Pr_{Z \sim \chi}[q_{Z} \geq 2^{-n/40}] \geq \gamma$.
First select randomness $r$ (to be fixed later in the algorithm), which will be used to generate an algorithm $A'_r$ that satisfies the conditions of \Cref{corollary:entropy}. Intuitively, we use $r$ to extract the inherent randomness of the quantum algorithm, and by fixing it later in the process, we achieve the desired low min-entropy property.
This algorithm will be nearly deterministic on a large fraction of possible inputs and goes as follows.
\begin{enumerate}
\item \label{step:1} Take $O((mn)^2 \cdot 2^{n/10})$ samples from $A$ in order to get empirical estimates $\tilde{p}_Z$ of  all of the probabilities $p_Z = \Pr_A[A^F\text{~Samples $Z$}]$, for all $Z \in (\{0,1\}^n)^m$, to within additive error $2^{-n/20}$ with high probability.
\item Pick a random threshold $\tau$ which is an integer multiple of $2^{-n/20}$ in the range $[\eps, 2\eps]$ for $\eps = 2^{-n/30}$.
\item Denote by $\mathcal{P} = \{Z: \tilde{p}_Z \ge \tau\}$ the set of ``probable'' outcomes. Consider the renormalized distribution $\tilde{p}_{Z| \mathcal{P}}$ over only the probable elements, i.e., for any $Z \in \mathcal{P}$ let $\tilde{p}_{Z| \mathcal{P}} = \frac{\tilde{p}_{Z}}{ \sum_{Z\in \mathcal{P}} p_Z}.$

\item \label{step:2} Use $r$ to rejection sample according to the probability distribution $\tilde{p}_{Z| \mathcal{P}}$.
Concretely, view $r$ as a list of uniformly random points $(\vec{x}_1,y_1), (\vec{x}_2,y_2), \dots \in (\{0,1\}^{n})^m \times [0,1]$. Output the first $\vec{x}_i$ such that $y_i$ is less than the estimated probability for $\vec{x}_i$.
We denote the output of this procedure by $\RejSamp(\tilde{p}_{Z|\mathcal{P}},r)$. Notice that since we fix $r$ later in the algorithm, $r$ (and the precision of $y_i$) can be as large as we need.
\end{enumerate}
We need to show that for some fixed choices of $(\tau, r)$, $A_{\tau,r}'$ is pseudo-deterministic and  succeeds $\mathsf{FHOG}$ with probability $\gamma/\poly(n)$. Then we can invoke \Cref{corollary:entropy}.

Fix an $F$ where the condition is not satisfied. First we notice that if $q_Z = \Pr[Z \sim \chi] \geq 2^{-n/40}$, then $\Pr_A[Z \sim A^F \wedge A^F \text{~solves~}\mathsf{FHOG} ] \geq 2^{-n/40} \cdot 1/\poly(n) \ge  3\cdot 2^{-n/30}$. 
This is clear since $\Pr_A[A^F \text{ passes } \mathsf{FHOG}] \geq \frac{1}{\poly(n)}$.
Thus, 
\begin{equation}\label{eq:tau}
\sum_{\substack{Z:p_Z\geq 3\cdot 2^{-n/30},\\ Z \text{ solves } \mathsf{FHOG} }} p_Z \ge \sum_{Z:p_Z\geq 3\cdot 2^{-n/30}} q_Z \cdot \delta(1/n) \ge  \gamma/\poly(n).
\end{equation}
So for any such $F$, we have $\frac{\gamma}{\poly(n)}$ probability mass on $Z: p_Z\geq 3\cdot 2^{-n/30}$ that also solve $\mathsf{FHOG}$.

Then, we note that by Chernoff bound the following event $\mathcal{E}$ happens with probability at least $1-\exp(-mn)$: for all $Z$, $|\tilde{p}_Z - p_Z| \le 2^{-n/20}$. 

Consider the process of rejection sampling over the true distribution $p$ and over the estimated distribution $\tilde{p}$.
Our first claim is: 
\begin{claim}\label{claim:st_dist}
Conditioned on $\mathcal{E}$, with $1-\exp(-\Omega(n))$ probability over the choice of $\tau$, the distributions $p_{Z|\mathcal{P}}$ and $\tilde{p}_{Z|\mathcal{P}}$ are $\exp(-\Omega(n))$-close in statistical distance.
\end{claim}
We defer the proof to Appendix \ref{ap:st_dist}.
Fix a ``good'' $\tau$ that satisfies that $p_{Z|\mathcal{P}}$ and $\tilde{p}_{Z|\mathcal{P}}$ are $\exp(-\Omega(n))$-close in statistical distance.
 This means that the output of Rejection Sampling will be identical for most choices of $r$ by the following claim.
\begin{claim}
If $\mathcal{D}$, $\mathcal{D}'$ are two distributions over the same finite domain $\mathcal{X}$ with statistical distance $\varepsilon$ then  $\Pr_{r}[\RejSamp(\mathcal{D},r) \neq \RejSamp(\mathcal{D}',r)]  = 2\varepsilon/(1+\varepsilon)$
\end{claim}
\begin{proof}
The probability that $\RejSamp(\mathcal{D},r)\neq \RejSamp(\mathcal{D}',r)$ is exactly the probability that we land in one of the ``bad'' regions $(\min\{\mathcal{D}(x),  \mathcal{D'}(x)\}, \max\{\mathcal{D}(x),  \mathcal{D'}(x)\}]$ for $x\in \mathcal{X}$ before reaching one the ``good'' region $[0,\min\{\mathcal{D}(x),  \mathcal{D'}(x)\}]$ for $x\in \mathcal{X}$.
Indeed, the overall area of the bad regions equals twice the statistical distance between $\mathcal{D}$ and $\mathcal{D}'$.
The overall area of the good regions equals $1-\eps$.
Thus the probability of reaching a bad region before reaching a good region is $2\eps/(2\eps + 1-\eps) = 2\eps/(1+\eps)$.
\end{proof}

Note that for \emph{any} $\tau$ whatsoever,
$$\sum_{\substack{Z:p_Z\geq \tau,\\ Z \text{ solves } \mathsf{FHOG} }} p_Z \ge \gamma/\poly(n)$$ using Equation~\eqref{eq:tau} and the fact that $\tau < 3\cdot 2^{-n/20}$.
Furthermore, under the event $\mathcal{E}$, we have
$$ 
\sum_{\substack{Z:\tilde{p}_Z\geq \tau,\\ Z \text{ solves } \mathsf{FHOG} }} \tilde{p}_Z \ge
\sum_{\substack{Z:p_Z\geq 3\cdot 2^{-n/30},\\ Z \text{ solves } \mathsf{FHOG} }} (p_Z-2^{-n/20}) \ge 
\gamma/\poly(n) - 2^{-\Omega(n)} \ge \gamma/\poly(n), $$
as well.

Overall we get that for any ``good'' $F$, with high probability $\mathcal{E}$ and $\tau$ is ``good''. Conditioned on $\mathcal{E}$ and $\tau$ is ``good'', we  have that $\RejSamp(\tilde{p}_{Z|\mathcal{P}},r) = \RejSamp(p_{Z|\mathcal{P}},r)$ with high probability.
Furthermore, under such conditioning $\RejSamp(\tilde{p}_{Z|\mathcal{P}},r)$  samples an element that passes $\mathsf{FHOG}$ with probability at least $\gamma/\poly(n)$.
We say that $(\tau,r)$ are deterministically-good if 
$$\Pr_{F \in \text{``Good Fs''},A}[\RejSamp(p_{Z|\mathcal{P}},r) \neq \RejSamp(\tilde{p}_{Z|\mathcal{P}},r)] \le \negl(1/n)$$
Indeed we see that most choices of $(\tau, r)$ are deterministically-good.
Let $$\delta := \E_{\tau, r}[\Pr_{F \in \text{``Good Fs''},A}[\RejSamp(\tilde{p}_{Z|\mathcal{P}},r) \text{ passes } \mathsf{FHOG}]] \ge \gamma /\poly(n)$$
We say that $(\tau,r)$ are $\mathsf{FHOG}$-good if 
$$\Pr_{F \in \text{``Good Fs''},A}[\RejSamp(\tilde{p}_{Z|\mathcal{P}},r) \text{ passes }\mathsf{FHOG}] \ge \delta/2.$$
By a simple Markov's inequality, we see that at least $\delta/2$ fraction of the choices of $(\tau,r)$ are $\mathsf{FHOG}$-good.
Thus, there exists a choice of $(\tau, r)$ that is both $\mathsf{FHOG}$-good, and deterministically-good. However, this leads to a contradiction of Corollary \ref{corollary:entropy}.
Indeed, we see that for this specific $(\tau,r)$ we have that at least $1/\poly(n)$ fractions of $F$ are ``Good'' in the sense that $A'_r$ outputs the constant value $\RejSamp(p_{Z|\mathcal{P}},r)$ with probability close to $1$, and at the same time $A'_r$ solves $\mathsf{FHOG}$ with probability at least $\gamma/\poly(n)$.
This shows that $\Pr[F \in \Good] \ge \gamma/\poly(n)$ which is a contradiction to Corollary \ref{corollary:entropy} for any $\gamma = 2^{-o(m)}$.
\end{proof}

\subsection{Success of the high-entropy algorithm}\label{subsec:he-algorithm}
We had better check that we can solve $\mathsf{FHOG}$ with a quantum algorithm with no restrictions on the min-entropy. Otherwise, our test would be impossible even for an honest algorithm! The simple quantum algorithm first prepares the following state using one query to the oracle for each $f_i$:
$$\frac{1}{\sqrt{N}} \sum_{x \in \{0, 1\}^n} f_i(x) \ket{x}$$
And then applies a Hadamard gate to all qubits and measures them in the computational basis. Here we verify that this usual high min-entropy algorithm solves $\mathsf{FHOG}$. Similar calculations can also be found in Lemma 8 of \cite{aaronsonbqpph}.
\begin{fact}
\label{fact:calcs}
If $A$ is the usual quantum algorithm for sampling from $\hat{f}_i$, then with probability $1-\exp(\Omega(m))$ over the choice of functions and the internal randomness of $A$:
$$|\mathcal{H}(F, A^F) - T_\mathcal{H}\cdot m | \leq c\cdot m, \qquad|\mathcal{L}(F, A^F) - T_\mathcal{L} \cdot m| \leq c \cdot m$$
\end{fact}
\begin{proof}
For each $z$, $\hat{f}(z)$ follows a normal distribution with mean 0 and variance $1/N$. So averaged over the choice of $f$, the probability that $A$ outputs a ``light'' outcome is given by:
\begin{align*}
\Pr[\hat{f}(A_f)^2\le 1/N] &= \E_f\left[\sum_{z : \hat{f}(z)^2 \le 1/N}\hat{f}(z)^2\right] = \sum_{z\in\{0,1\}^n} \E_f\left[ \hat{f}(z)^2 \cdot \mathds{1}_{\{\hat{f}(z)^2 \le 1/N\}} \right] \\
&\approx N \cdot \E_{X\sim \mathcal{N}(0,1/N)}[X^2 \cdot \mathds{1}_{\{X^2\le 1/N\}}] 
= \E_{Y\sim \mathcal{N}(0,1)}[Y^2 \cdot \mathds{1}_{\{Y^2\le 1\}}]
= T_\mathcal{L}.
\end{align*}
Similarly,
\begin{align*}
\Pr[1/N<\hat{f}(A_f)^2\le 4/N] &\approx  \E_{Y\sim \mathcal{N}(0,1)}[Y^2 \cdot \mathds{1}_{\{1<Y^2\le 4\}}] = T_\mathcal{H}
\end{align*}
Since $A$ repeats this $m$ times, a Chernoff bound followed by a union bound proves the claim.
\end{proof}
\subsection{Proof of min-entropy}
To specify a certified randomness protocol we consider a setting in which a prover can certify to a verifier that a bit string has been sampled from a distribution that is statistically close to uniformly random. We demand that the protocol is complete -- there exists a prover that fails with non-negligible probability over the random $F$, and is sound -- no dishonest prover can force the verifier to accept a heavily biased output with non-negligible probability. This is easily achievable if the verifier can access a high min-entropy source, since it can always use classical randomness extractors \cite{guv} to extract uniformly random bits.
Theorem \ref{thm:linear-entropy} and Fact \ref{fact:calcs} give us the tools to construct a proof of min-entropy protocol. Let $m = \text{poly}(n)$, and consider the following protocol:
\begin{enumerate}
    \item Given $F$, the prover attempts to pass $\mathsf{FHOG}$, and sends $Z = z_1, \ldots, z_m$ as proof.
    \item The verifier computes $|\mathcal{H}(F, Z)|$ and $|\mathcal{L}(F, Z)|$ (using exponential number of queries), and accepts if they are within the expected range.
\end{enumerate}

To show completeness, we use Fact \ref{fact:calcs}. This effectively shows that there exists a quantum algorithm that runs in polynomial time and finds a good solution that the verifier rejects with negligible probability.

The soundness follows from Theorem \ref{thm:linear-entropy}, which states that with high probability over the choice of $F$, the output distribution of any $\poly(n)$ query quantum algorithm -- conditioned on the verifier accepting -- is $\negl(n)$ close in total variation distance to a $\Omega(n)$ min-entropy distribution.
\newline

Note that this is similar to the definition of the min-entropy protocol defined in \cite{YK22} with two main differences. First, the verification is not efficient and requires an exponential number of queries to the oracle. Second, the soundness argument in \cite{YK22} requires that the output distribution must have high min-entropy with high probability over the choice of $F$. Our soundness requirement is less stringent and allows the distribution to be statistically close to a high min-entropy distribution. Using randomness extractors on a distribution that is close in total variation distance to a high min-entropy distribution would only affect the closeness of the random bits to the uniform distribution. As long as this statistical distance is negligible this minor difference does not affect the protocol.

\section{Black Box \texorpdfstring{$\mathsf{LLQSV}$}{llqsv}} \label{sec:maj}
So far we focused on how Fourier Sampling can independently be seen as a tool for certified random number generation in the QROM. In \cite{a-talk1}, Aaronson considers certified randomness in the white-box setting and is able to specify a protocol assuming the $\llqsv$ conjecture. To give evidence for $\llqsv$ in the following two sections we consider black-box variants of the $\llqsv$ conjecture. 

First, we formally define the black-box variant of $\llqsv$:
\begin{problem} [$\bllqsv$]
Given oracle access to $\mathcal{O}$, a list of pairs of functions $f_i: \{0, 1\}^n \rightarrow \{\pm 1\}$ and outcomes $s_i \in \{0, 1\}^n$ of length $m$, distinguish whether $\mathcal{O}$ was sampled according to distribution $\mathcal{U}$ or $\mathcal{D}$, where:
 \begin{itemize}
  \item $\mathcal{U}$: For each $i \in [m]$, $f_i$ and $s_i$ are chosen uniformly at random.
  \item $\mathcal{D}:$ For each $i \in [m]$, $f_i$ is a random Boolean function, and $s_i$ is sampled according to $\hat{f}_i^2$.
 \end{itemize}
\end{problem}
 
We observe there is a close connection between $\bllqsv$ and the $\mathsf{MAJORITY}$ problem. To see this we notice that in either case the probability of sampling a pair $(f, s)$ only depends on the magnitude of the Fourier coefficient of $f$ at $s$, or equivalently $|h(f\cdot\chi_s)-N/2|$, where $h$ is the hamming weight and $\chi_s$ is the $s$-Fourier character, i.e., $\chi_s(x) = (-1)^{x\cdot s})$ (For simplicity we also refer to the $2^n$ bit string generated by the truth table of functions as $f$).
This allows us to write both distributions as a linear combination of simpler distributions $\mathcal{U}_d^N$ -- the uniform distribution over all strings of hamming weight $N/2+d$ or $N/2-d$. 
To provide evidence for $\mathsf{QCAM}$ hardness, we unconditionally prove lower bounds against $\mathsf{BQP}$ and $\mathsf{PH}$ (And hence, against $\mathsf{AM}$ \cite{BHZ87}). 

First, we notice an alternative way to sample from the distribution $\mathcal{D}$ that makes our analysis easier. Since $\forall z\in \{0, 1\}^n: \Pr_\mathcal{D}[s_i = z] = \frac{1}{N}$, we can sample from the joint distribution of a random function and a sample according to the Fourier distribution by first sampling a uniformly random outcome $s_i$, and then picking a random function $f_i$ weighted according to $\hat{f}_i(s_i)^2$. Both of our results in this section follow two observations. Let $\mathcal{U}_d^N$ be the uniform distribution over the following set:
 $$\{S \in \{0, 1\}^N | h(S) = N/2 - d \lor h(S) = N/2 + d\}$$
 First, for both distributions $\mathcal{U}$ and $\mathcal{D}$ the distribution over $f_i \cdot \chi_{s_i}$ is a linear combination of $\mathcal{U}_0^N, \mathcal{U}_1^N, \ldots, \mathcal{U}_{N/2}^N$. In the uniform case, this is clear since $f_i \cdot \chi_{s_i}$ is a uniformly random string. For $\mathcal{D}$, hamming weight of $f_i \cdot \chi_{s_i}$ specifies the Fourier coefficient $\hat{f}_i(s_i)$, so all strings with hamming weight $N/2 + d$ or $N/2 - d$ occur with equal probability. So, we can write both probability distributions as:
 $$\sum_{D\coloneqq (d_1, \ldots, d_m) \in [N/2+1]^m} p_D \mathcal{U}_{d_1} \times \ldots \times \mathcal{U}_{d_m} $$
 For some set of coefficients where $\sum_D p_D = 1$. Let us call these distributions $\chi_\mathcal{D}$ and $\chi_\mathcal{U}$ respectively. 
 
 Next, we use the fact that in both distributions $\mathcal{U}$ and $\mathcal{D}$ the marginal distribution over the functions is uniform. Thus, with high probability, the Fourier coefficients are bounded:

 \begin{theorem} \label{th:typical}
  Given a random Boolean function $f: \{\pm 1\}^n \rightarrow \{0, 1\}$, the probability that a squared Fourier coefficient of $f$ is larger than $\frac{p(n)^2}{N}$ is at most $2 \exp(\frac{-p(n)^2}{6} + \ln N)$.
 \end{theorem}

 \begin{proof}
  This follows from a Chernoff bound. Let $X = h(f\cdot \chi_s)$. For every $s$, this is a uniformly random string, so we can write:
  $$\Pr\left[X \leq \left(1 - \frac{p(n)}{\sqrt{N}}\right) N/2\right] \leq \exp(-p(n)^2/4) $$
  $$\Pr\left[X \geq \left(1 + \frac{p(n)}{\sqrt{N}}\right) N/2\right] \leq \exp(-p(n)^2/6) $$
  The claim follows from union bound.
 \end{proof}

 Note that in both $\mathcal{D}$ and $\mathcal{U}$, $f_i$ is a random Boolean function. So in either case by a union bound we can prove that with probability at least $1 - 2\exp(\frac{-p(n)^2}{6} + \ln N + \ln m)$ none of these functions have a squared Fourier coefficient larger than $\frac{p(n)^2}{N}$. So as long as $\ln m$ is polynomial in $n$ we can choose $p(n)$ in a way that this event happens with an exponentially small probability. 

It is clear that distinguishing $\mathcal{D}$ and $\mathcal{U}$ is harder than distinguishing $\chi_\mathcal{U}$ and $\chi_\mathcal{D}$ since given $f_i, s_i$ we can compute $f_i \cdot \chi_{s_i}$. The other direction is not as clear since given a long list of strings $S_1, \ldots, S_m$, we need $s_1, \ldots, s_m$ to recover $f_1, \ldots, f_m$. However, given a sample from $\chi_\mathcal{D}$ or $\chi_\mathcal{U}$, we can easily get a sample from $\mathcal{D}$ or $\mathcal{U}$ respectively, by simply choosing uniformly random $s_1, \ldots, s_m$. So, assuming we have access to a long list of random bits, any distinguisher for $\chi_\mathcal{U}$ and $\chi_\mathcal{D}$ also distinguishes $\mathcal{D}$ and $\mathcal{U}$. As we will see in analysis having access to an extra random oracle does not give any extra computation power to the $\BQP$ or $\PH$ algorithm. Combining Theorem \ref{th:typical} with this fact implies that any distinguisher for $\mathcal{U}$ and $\mathcal{D}$ must also distinguish the following two distributions with high probability:
 \begin{itemize}
  \item $\chi_\mathcal{U}'$ = $\sum_{d_1 \ldots d_m \in [p(n)^2 \sqrt{N}]} p_D \mathcal{U}_{d_1}^N \times \ldots \times \mathcal{U}_{d_m}^N$
  \item $\chi_\mathcal{D}'$ = $\sum_{d_1 \ldots d_m \in [p(n)^2 \sqrt{N}]} q_D \mathcal{U}_{d_1}^N \times \ldots \times \mathcal{U}_{d_m}^N$
 \end{itemize}
 Since $\chi_\mathcal{U}$ and $\chi_\mathcal{D}$ are statistically close to $\chi_\mathcal{U}'$ and $\chi_\mathcal{D}'$ respectively. Furthermore, since $\sum_D p_D \leq 1$ and $\sum_D q_D \leq 1$, by an averaging argument this distinguisher can also distinguish two distributions in the sum with a $1/\poly(n)$ advantage. This reduces the problem to indistinguishability of  $\mathcal{U_D}$ and $\mathcal{U}_{D'}$, where $d_i$ and $d'_i$ are at most $p(n) \sqrt{N}$. Thus, to complete the proof we need to show $\mathsf{BQP}$ and $\mathsf{PH}$ lower bounds for the following problem:
 \begin{problem} [$\distbal$]\label{prob:dist-bal}
 Let $0 \leq d_i \leq p(n) \cdot \sqrt{N}$. Given oracle access to a list of strings $S = S_1, \ldots, S_m$, where $S_i \in \{0, 1\}^N$, distinguish whether the oracle was sampled according to distribution $\mathcal{U}_D$ or $\mathcal{U}_0$, where:
 \begin{enumerate}
     \item $\mathcal{U}_D$: $S_i$ is drawn from $\mathcal{U}_{d_i}^N$.
    \item $\mathcal{U}_0$: $S_i$ is drawn from $\mathcal{U}_{0}^N$.
 \end{enumerate}
 \end{problem}
Note that indistinguishability of $\mathcal{U}_D$ and $\mathcal{U}_0$ implies indistinguishability of $\mathcal{U}_D$ and $\mathcal{U}_{D'}$. We say algorithm $A$ solves $\distbal$ with advantage $\delta$, if:
 $$\Big|\Pr_{S \sim \mathcal{U}_D}[A \text{ accepts } S] - \Pr_{S \sim \mathcal{U}_0}[A \text{ accepts } S]\Big| \geq \delta $$
 
In our proofs, we use a hybrid argument to get the $\mathsf{BQP}$ lower bound. Furthermore, this reduction enables us to use an argument similar to what is given in \cite{f-beat, aaronsonbqpph} to prove a bound against $\mathsf{AC^0}$. Informally, since this simplified problem is random self-reducible, combined with a circuit lower bound due to Håstad we can beat the hybrid argument, and extend the lower bound to an exponentially long list. Notice that black-box $\llqsv$ is slightly different than $\distbal$, since it also has access to a random string for each $S_i$ (which is the outcome $s_i$). To complete our proof, we need to show that having access to these extra random strings do not give us any extra computation power. This is straightforward from the hybrid argument for the $\mathsf{BQP}$ lower bound, and we can again use random self-reducibility of $\distbal$ to show the same claim for the $\mathsf{PH}$ argument.

 \subsection{\texorpdfstring{$\mathsf{BQP}$}{bqp} Lower Bound for \texorpdfstring{black-box $\llqsv$}{black-box llqsv}}\label{subsec:bqp-lowerbound}

 \begin{theorem} \label{th:bqpll}
  Any $T$ query quantum algorithm for $\distbal$ has advantage at most $p(n)\cdot TN^{-1/8}$.
 \end{theorem}
 \begin{proof}
  This follows from a BBBV style argument. Suppose that $S$ is drawn from $\mathcal{U}_{d_1}^N \times \ldots \times \mathcal{U}_{d_m}^N$. We can modify each $S_i$ in at most $p(n) \cdot \sqrt{N}$ entries to make each $S_i$ balanced. Let this new string be $S'$. Let $P_i$ be the set of possible entries for each string $S_i$, and $P = \bigcup_i P_i$. Then given a quantum algorithm $A$, we can write the query magnitude of each $P_i$ as:
  $$\sum_{t=1}^T \sum_{\substack{w \\ x \in P_i}} |\alpha_{t, x, w}|^2 \leq T_i $$
  and $\sum_{i} T_i = T$. So if we choose a random subset $\Delta = \bigcup_i \Delta_i$ for the modified entries of each $S_i$, we can write the expected value of the query magnitude as:
  $$\E_\Delta \Big[ \sum_{t=1}^T \sum_i^m \sum_{\substack{w \\ x \in \Delta_i}} |\alpha_{t, x, w}|^2 \Big] = \sum_i^m \E_{\Delta_i} \Big[ \sum_{t=1}^T \sum_{\substack{w \\ x \in \Delta_i}} |\alpha_{t, x, w}|^2\Big] \leq \sum_i^m \frac{T_i p(n) \sqrt{N}}{2|P_i|} \leq \frac{p(n)T}{\sqrt{N}} $$
  And by Markov's inequality and Cauchy-Schwartz, we get that for $1-N^{-1/4}$ fraction of $S'$:
  $$\delta(A(S), A(S')) \leq 4\cdot p(n)\cdot TN^{-1/8}$$
  However, this random modification also preserves the distribution in this case. To sample according to $\mathcal{U}_0$, we can first sample from $\mathcal{U}_d$ and then randomly flip the excess bits to make the string balanced. So the advantage of $A$ is at most:
  $$N^{-1/4} + 4\cdot p(n)\cdot TN^{-1/8} \leq 8 \cdot p(n) \cdot TN^{-1/8} $$
  
 \end{proof}
 \begin{theorem} [$\mathsf{BQP}$ Bound]
Any quantum algorithm for $\bllqsv$ with advantage $\delta=1/\poly(n)$ must make at least $\Omega(N^{1/10})$ queries to the oracle.
 \end{theorem}
 \begin{proof}
 Suppose there exists a $O(N^{1/10})$ quantum query algorithm for black-box $\llqsv$. First we notice that the same hybrid argument works without any modification if in $\distbal$ we have access to $(S_i, s_i)$, where $s_i$ is a uniformly random string. Now we can use $s_i$ to recover $f$ by computing $S_i \cdot \chi_{s_i}$. From Theorem \ref{th:typical}, any distinguisher for black-box $\llqsv$ must also solve $\distbal$ (with the additional access to $s_i$) with advantage at least $\delta'=\delta-2\exp(\frac{-p(n)^2}{6}-\ln N - \ln m)$. We can choose an appropriate polynomial $p$ so that $\frac{p(n)^2}{6}-\ln N - \ln m \geq n$. Then for large enough $n$, this implies $\delta'\geq \delta/2 = 1/\poly(n)$. However, we know that:
 $$\delta' \leq 8\cdot p(n) \cdot cN^{1/10} \cdot N^{-1/8} = O(N^{-1/40}) $$
 Which gives us a contradiction.
 \end{proof}

\subsection{\texorpdfstring{$\mathsf{PH}$}{PH} Lower Bound for \texorpdfstring{black-box $\llqsv$}{llfs}}\label{subsec:secondph-lowerbound}  
We start from a well known $\mathsf{AC^0}$ lower bound~\cite{has}:
\begin{lemma} \label{lemma:has}
Given an $N$ bit string $S$, any depth $d$ circuit that accepts all $S$ such that $h(S)=N/2+1$ and rejects them if $h(S)=N/2$ has size: $$\exp\big(\Omega(N^{1/(d-1)}\big)$$
\end{lemma}
First we notice that the same bound holds for distinguishing $h(S) = N/2$ from $|h(S) - N/2| = 1$. The next step is to prove analog of Theorem \ref{th:bqpll} against $\mathsf{PH}$.

\begin{theorem} \label{th:distbal-ac}
Any depth $d$ circuit that solves $\distbal$ with $1/\poly(n)$ advantage has size:
$$\exp\big(\Omega(N^{1/(2d+4)})\big)$$
\end{theorem}
\begin{proof}
Suppose there exists a circuit $C$ of size $s$ and depth $d$ that solves $\distbal$ with advantage $\delta = 1/\poly(n)$. Using $C$ we can construct a circuit $C'$ of size $s'=\poly(s, N)$ and depth $d+3$ that solves the problem in Lemma \ref{lemma:has} for length $\sqrt{N}/p(n)$. Given a $\sqrt{N}/p(n)$ bit string $t$, for each $i$, we copy the string $d_i$ times and then pad it with an equal number of zeros and ones to have length $N$. We then take a random permutation of this string and flip all bits with probability $1/2$ to get $S_i$. Note that if the original string is balanced $S_i \sim \mathcal{U}_0^N$, and otherwise $S_i \sim \mathcal{U}_D^N$. Then we know:
$$\Big|\Pr\big[A \text{ accepts } S | t \text{ balanced} \big]-\Pr\big[A \text{ accepts } S | t \text{ not balanced}]\Big| \geq \delta$$
We can repeat this $r=\poly(1/\delta, N)$ times. By Chernoff bound, we get a $1/\delta$ gap with arbitrary large probability $\exp\big(-\poly(N)\big)$, which is detectable by an approximate majority circuit~(see \cite{aaronsonbqpph, DBLP:conf/coco/Viola07} for more detail). Let us call this new circuit $C'$. We can then fix the randomness in a way that the new circuit $C'$ succeeds on every input simultaneously. This new circuit has depth $d+3$ and size $\poly(s, N)$. So $s$ must be at least:
$$\exp\big(\Omega(\sqrt{N}^{-1/(d+2)} \big) = \exp\big(\Omega(N^{1/(2d+4)} \big)$$
\end{proof}
\begin{corollary} \label{cor:2}
Any depth $d$ circuit that solves black-box $\llqsv$ with $1/\poly(n)$ advantage has size:
$$\exp\big(\Omega(N^{1/(2d+4)})\big)$$
\end{corollary}
 \begin{proof}
Similar to the $\mathsf{BQP}$ bound, the first step is to prove the same $\mathsf{AC^0}$ bound for $\distbal$ even when the algorithm has additional access to a long list of random outcomes $s_1, \ldots, s_m$. This immediately follows from the proof idea of Theorem \ref{th:distbal-ac}, since we can choose uniformly random $s_i$ in the circuit for each instance separately (without accessing the oracle, while preserving the distribution), and later fix them in the circuit. Furthermore, from Theorem \ref{th:typical}, any $1/\poly(n)$ distinguisher for $\bllqsv$ would also work for $\distbal$ with the additional outcome oracle by choosing the polynomial $p$ to be large enough.
\end{proof}
Corollary \ref{cor:2} and the standard conversion between $\mathsf{PH}$ and $\mathsf{AC^0}$ proves a $\mathsf{PH}$ lower bound for solving $\bllqsv$.

\section{The \texorpdfstring{$\sqfor$}{Square Forrelation} Problem and \texorpdfstring{$\bllqsv$}{Black-Box LLQSV}} \label{sec:bqpph}
In this section, we first prove an alternative oracle separation between $\mathsf{BQP}$ and $\PH$, and then show a reduction from $\sqfor$ to a problem similar to $\bllqsv$. To prove the lower bound we build on the previous work of Raz and Tal~\cite{DBLP:conf/stoc/RazT19} to define a distribution $\mathcal{D}$ over pairs of Boolean functions $(f, g)$ that is indistinguishable from uniform in the polynomial hierarchy, and yet a simple quantum algorithm is able to distinguish them. 

This hardness result does not immediately follow from previous oracle separations and requires careful modifications of Raz and Tal analysis to make them work with squared Fourier coefficients for both the $\mathsf{PH}$ lower bound and the quantum algorithm. Furthermore, we extend these results to the ``long list'' setting and prove that having oracle access to an exponentially long list of samples does not make the problem any easier. 

Before getting into the formal results, let us start with the simplified \cref{prob:simplifiedsqfor}. One can choose the threshold parameter to be the median of squared Fourier coefficients of a random Boolean function, so that almost half of the coefficients are indicated as heavy on average, i.e\ $\Pr_{x, f}[g(x)=1] \approx 1/2$.

Then one can easily find random heavy Fourier coefficients by finding a random $x$ such that $g(x) = 1$. The distribution $\mathcal{D}$ that we define next is similar to simplified $\sqfor$ in the sense that we use randomized rounding on squared Fourier coefficients and we can prove that on average, heavy Fourier coefficients have a slightly higher chance of having $g(x) = 1$. We leverage this to reduce $\sqfor$ to a black-box variant of $\mathsf{LLQSV}$. 

Let us first summarize the separation result that we are going to use to give evidence for $\mathsf{LLQSV}$. Similar to \cite{DBLP:conf/stoc/RazT19}, to show an oracle separation of $\mathsf{BQP}$ and $\mathsf{PH}$ it suffices to find a distribution $\mathcal{D}$ that is pseudorandom for $\mathsf{AC^0}$ circuits, but not for efficient quantum algorithms making few queries. First, we start by defining distribution $\Dd$ using truncated multivariate Gaussians. Let $n \in \mathbb{N}$, $N = 2^n$. Let $\varepsilon = 1/(C \ln N)$ for a constant $C\ge 20$ to be chosen later. Define $\mathcal{G}$ to be a multivariate Gaussian distribution over $\mathbb{R}^N \times \mathbb{R}^N$ with mean $0$ and covariance matrix:
\begin{equation*}
\varepsilon \cdot \begin{pmatrix}
I_N  & H_N  \\
H_N & I_N
\end{pmatrix}
\end{equation*}
Where $H_N$ is the Hadamard transform. To take samples from $\mathcal{G}$, we can first sample $X = x_1, \ldots, x_N \sim \Nn(0, \varepsilon)$, and let $Y = H_N \cdot X$. Define $\mathcal{G}'$ to be the distribution over $Z = (X, Y^2 - \varepsilon)$. Let $\trunc(a) = \min(1, \max(-1, a))$. The distribution $\mathcal{D}$ over $\{\pm 1\}^{2N}$, first draws $Z \sim \mathcal{G}'$. Then for each $i \in [2N]$ draws $z'_i = 1$ with probability $\frac{1+\trunc(z_i)}{2}$ and $z'_i=-1$ with probability $\frac{1-\trunc(z_i)}{2}$. Now we can define the following promise problem:
\begin{problem} [$\sqfor$]
Given oracle access to Boolean functions $f, g: \{0, 1\}^n \rightarrow \{\pm 1\}$, distinguish whether they are sampled according to $\mathcal{D}$ or uniformly at random.
\end{problem}
Recall that we can write the multilinear expansion of a Boolean function $F: \mathbb{R}^{2N} \rightarrow \mathbb{R}$ as:
$$F(z) = \sum_{S \subseteq [2N]} \hat{F}(S) \prod_{i \in S}z_i $$
Similar to the original proof, it is not hard to see that the randomized rounding step does not change the expected value of the outcome for multilinear functions:
\begin{align*} 
  \E_{z \sim \mathcal{G'}} [F(\trnc(z))] = \E_{z' \sim \mathcal{D}}[F(z')]
\end{align*}
Hence we can ignore step 3 in our analysis. More importantly, because of the choice of $\varepsilon$, truncations happen with negligible probability, and in the event that they happen, it is possible to bound the difference. 
$$  \E_{z \sim \mathcal{G'}} [F(\trnc(z))]  \approx \E_{z \sim \mathcal{G'}}[F(z)]$$
The rest of the analysis is to prove indistinguishability of $(X, Y^2 - \varepsilon)$ from the uniform distribution. More specifically, to use the tail bound from~\cite{Tal17} for the Fourier coefficients of bounded depth circuits to bound:
$$\left|\E[F(X, Y^2-\varepsilon)] - \E[F(\Uu_{2N})]\right| = \left|\E[F(X, Y^2 - \varepsilon) - F(0,0)]\right| $$
The idea is to use the same telescopic sum as Raz and Tal for analyzing the expected value of $F$ and think of $\mathcal{G'}$ as a Brownian motion, but also consider the squared inputs and the $\varepsilon$ difference. One important step is to also break the $\varepsilon$ bias in the telescopic sum. 
Let $t \in \mathbb{N}$, and for $i \in [t]$, let $X^{(i)} = X^{(i-1)} +\Delta_X^{(i)}$ and $Y^{(i)} = H_N \cdot X^{(i)}$, where $\Delta_X^{(i)} \sim \mathcal{N}(0, \varepsilon/t)$. Now we can write:
\begin{align*}
  F(X^{(t)}, (Y^{(t)})^2 - \varepsilon) - F(0,0) = \sum_{i=1}^t \Bigg\{F\left(X^{(i)}, (Y^{(i)})^2 - \tfrac{\varepsilon\cdot i}{t}\right) - F\left(X^{(i-1)}, (Y^{(i-1)})^2 - \tfrac{\varepsilon\cdot (i-1)}{t}\right)\Bigg\}
\end{align*}
Then we can use known bounds for each term in the sum and the triangle inequality to bound the sum. In the limit of $t \rightarrow \infty$, we are dealing with a stochastic integral so one can also prove this result using stochastic calculus.
However, we decided not to present the proof using the language of stochastic calculus and instead picked $t$ to be a large enough polynomial in $N$ (similar to \cite{DBLP:conf/stoc/RazT19}).

\subsection{Truncated Gaussians} \label{sec::trunc}
It is not hard to see that any multilinear function $F: \mathbb{R}^{2N} \rightarrow \mathbb{R}$ still has similar expectation under $\mathcal{D}$ and $\mathcal{G}'$, where:
$$F(z)=\sum_{S \subseteq[2 N]} \widehat{F}(S) \cdot \prod_{i \in S} z_{i}$$
First, we need to show that:
$$\E_{z^{\prime} \sim \mathcal{D}}\left[F\left(z^{\prime}\right)\right]=\E_{z \sim \mathcal{G}^{\prime}}[F(\operatorname{trnc}(z))]$$
And since we can still prove that the truncations happen with negligible probability, we only need to adapt the rest of the proof to work with squared Fourier coefficients.
The proof of the first part works without any modifications from the definition of distributions $\mathcal{D}$ and $\mathcal{G}'$:
\begin{align} 
\E\left[F\left(z^{\prime}\right) \mid z\right] &=\E\left[\sum_{S \subseteq[2 N]} \widehat{F}(S) \cdot \prod_{i \in S} z_{i}^{\prime} \;\middle|\; z\right]=\sum_{S \subseteq[2 N]} \widehat{F}(S) \cdot \prod_{i \in S} \E\left[z_{i}^{\prime} \;\middle|\; z\right] \nonumber \\
&=\sum_{S \subseteq[2 N]} \widehat{F}(S) \cdot \prod_{i \in S} \operatorname{trnc}\left(z_{i}\right)=F(\operatorname{trnc}(z)) \label{eq1}
\end{align}
The next fact is stated as Claim 5.1 in \cite{DBLP:conf/stoc/RazT19}, and is unaffected by our modification.
\begin{fact}
  Let $F: \mathbb{R}^{2N} \rightarrow \mathbb{R}$ be a multilinear function that maps $\{\pm 1\}^{2N}$ to $[-1, 1]$. Let $z = (z_1, \ldots, z_{2N}) \in \mathbb{R}^{2N}$. Then, $|F(z)| \leq \prod_{i=1}^{2 N} \max \left(1,\left|z_{i}\right|\right)$.
\end{fact}
The next two theorems are equivalent to Claim 5.2 and Claim 5.3 of \cite{DBLP:conf/stoc/RazT19} respectively, where they are proven over distribution $\mathcal{G}$. One can verify that they also hold over distribution $\mathcal{G'}$. In fact, by choosing $\varepsilon$ to be small enough, they even hold over $\mathcal{G'}^{\poly(N)}$, where we are given oracle access to an exponentially long list of samples from $\mathcal{G'}$. This is crucial for extending this $\mathsf{PH}$ bound to the long list problem. We provide their proofs in Appendix \ref{app:trunc} for completeness.

\begin{theorem}\label{th::5.2}
For any constant $c\ge 2$ there exists a choice of a constant $C$ in the definition of $\mathcal{G}^{\prime}$ such that:
  $\E_{(x, y) \sim \mathcal{G}^{\prime}}\left[\prod_{i=1}^{N} \max \left(1,\left|x_{i}\right|\right) \cdot \prod_{i=1}^{N} \max \left(1,\left|y_{i}\right|\right) \cdot \mathbbm{1}_{(x, y) \neq \operatorname{trnc}(x, y)}\right] \leq 4 \cdot N^{-c}$.
\end{theorem}

\begin{theorem}\label{th::5.3}
For any constant $c \ge 2$ there exists a choice of a constant $C$ in the definition of $\mathcal{G}^{\prime}$ such that the following holds. Let $0 \leq p, p_0$ such that $p+p_0 \leq 1$. Let $F: \mathbb{R}^{2N} \rightarrow \mathbb{R}$ be a multilinear function that maps $\{\pm 1\}^{2N}$ to $[-1, 1]$. Let $z_0 \in [-p_0, p_0]^{2N}$. Then,
  $$\E_{z \sim \mathcal{G}^{\prime}}\left[\left|F\left(\operatorname{trnc}\left(z_{0}+p \cdot z\right)\right)-F\left(z_{0}+p \cdot z\right)\right|\right] \leq 8 \cdot N^{-c}$$
\end{theorem}

\subsection{Quantum Algorithm for Distinguishing \texorpdfstring{$\mathcal{D}$}{D} and \texorpdfstring{$U_{2N}$}{U2n}} \label{sec::algorithm}
Here, for completeness, we provide a quantum algorithm for distinguishing $U_{2N}$ and $\mathcal{D}$. Note that this quantum algorithm does not extend to solve $\bllqsv$, and does not contradict the $\mathsf{LLQSV}$ conjecture. As we will show, a problem similar to $\bllqsv$ is harder than $\sqfor$. However, we use these calculations in the next sections.
Given to Boolean functions $f,g: \{\pm 1\}^N$ we use the following algorithm $A$ to distinguish between functions sampled from $\mathcal{D}$ and $U_{2N}$:
\begin{enumerate}
  \item Apply Fourier transform on $f$ and sample $x$ from the distribution induced by $\hat{f}$.
  \item Accept if $g(x) = 1$, and reject otherwise.
\end{enumerate}
We can see that the success probability of this algorithm is given by $\sum_{x, g(x) = 1} \hat{f}(x)^2 = \frac{1+\varphi(f, g)}{2}$, where $\varphi(f, g) = \varphi(X, Y)$ is defined as follows:
$$\varphi(X, Y) = \frac{1}{N}\sum_i (\sum_j H_{ij}X_i)^2 Y_i = \Pr[A \text{ Accepts}] - \Pr[A \text{ Rejects}]$$
We first analyze this quantity for when the samples are taken from $\mathcal{G'}$. In fact, we can show that the multilinear part of $\varphi(X, Y)$, denoted by $\varphi_{i \neq j}(X, Y)$ is large on average. Lastly, we can use Theorem \ref{th::5.3} to show that a similar bound holds for samples from $\mathcal{D}$ while the quadratic part of $\varphi(X, Y)$ is $0$ on average. To summarize, we show that:
$$  \E_{(X, Y) \sim \mathcal{D}}[\varphi(X, Y)] = \E_{(X, Y) \sim \mathcal{D}}[\varphi_{i \neq j}(X, Y)]\approx \E_{(X, Y) \sim \mathcal{G}'}[\varphi_{i \neq j}(X, Y)] = O(\varepsilon^2)$$

\begin{theorem}
  $\E_{(X, Y) \sim U_{2N}}[\varphi(X, Y)] = 0$.
\end{theorem}
\begin{proof}
  For every $i,j,k \in [N]$, $\E_{(X,Y)\sim U_{2N}}[X_j X_k Y_i] = 0$ because $Y_i$ is independent of $X$ and has mean $0$. By linearity of expectation $\E_{(X, Y) \sim U_{2N} }[\varphi(X, Y)] = 0$.
\end{proof}
\begin{theorem}
  $\E_{(X, Y') \sim \mathcal{G'}}[\varphi_{i \neq j}(X, Y')] = \varepsilon^2\cdot (2-2/N)$.
\end{theorem}
\begin{proof}
By the definition of $\mathcal{G'}$ and $\varphi_{i \neq j}(X, Y')$ we can write:
\begin{align*}
  \E_{(X, Y') \sim \mathcal{G'}}[\varphi_{i\neq j}(X, Y')] &= \frac{1}{N} \sum_i \sum_{j\neq k} H_{ij} H_{ik} \E[X_j X_k Y'_i] \\
  &= \frac{1}{N} \sum_i \sum_{j\neq k} H_{ij} H_{ik} \E[X_j X_k Y_i^2] - \varepsilon \E[X_j X_k] \\
  &= \frac{1}{N} \sum_i \sum_{j \neq k} H_{ij} H_{ik} (\sigma'_{ii} \sigma'_{jk} + 2 \sigma_{ij} \sigma_{ik}) \\
  &= \frac{1}{N} \sum_i \sum_{j \neq k} H_{ij}^2 H_{ik}^2 2 \varepsilon^2 = \frac{2(N-1)}{N} \varepsilon^2  
\end{align*}
The third line follows from the fact that $(X, Y)$ is a multivariate normal distribution, and the last inequality holds for $N\geq 2$. Furthermore, $\sigma'_{ij}$ is the covariance of $(X_i, X_j)$ and $(Y_i, Y_j)$, and $\sigma_{ij}$ is the covariance of $(Y_i, X_j)$.
\end{proof}
\begin{theorem} \label{th::qa}
  $\E_{(X, Y) \sim \mathcal{D}}[\varphi(X, Y)] \geq \varepsilon^2$.
\end{theorem}
\begin{proof}
  Since $\E[Y_i] = 0$, by linearity of expectation we can write:
  \begin{align*} 
   \E_{(X, Y) \sim \mathcal{D}}[\varphi(X, Y)]  &= \frac{1}{N} \sum_i \sum_{j \neq k} H_{ij} H_{ik} \E[X_j X_k Y_i] + \frac{1}{N} \sum_i \sum_{j} H_{ij}^2 \E[X_j^2 Y_i] \\
   & =  \frac{1}{N} \sum_i \sum_{j \neq k} H_{ij} H_{ik} \E[X_j X_k Y_i] + \frac{1}{N} \sum_i \sum_{j} H_{ij}^2 \E[Y_i] \\
   & = \E_{(X, Y) \sim \mathcal{D}}[\varphi_{i \neq j}(X, Y)] 
  \end{align*} 
  Note that $\varphi_{i \neq j}(X, Y)$ is a multilinear function that takes $\{\pm 1\}^{2N}$ to $[-2, 2]$ since the subtracted quadratic part maps $\{\pm\}^{2N}$ to $[-1, 1]$ and $\varphi(X, Y) \in [-1, 1]$. So we can use Theorem \ref{th::5.3} with $p_0 = 0, p = 1$ and Equation \ref{eq1} to write:
  \begin{align*} 
    \left|\E_{(X, Y) \sim \mathcal{D}}[\varphi_{i \neq j}(X, Y)] - \E_{(X, Y) \sim \mathcal{G'}} [\varphi_{i \neq j}(X, Y)]\right| \leq 16 \cdot N^{-2}
  \end{align*}
  Then we can get a lower bound on $\E_{(X, Y) \sim \mathcal{D}}[\varphi(X, Y)]$:
  \begin{equation*}
 \E_{(X, Y) \sim \mathcal{D}}[\varphi_{i \neq j}(X, Y)] \geq \E_{(X, Y) \sim \mathcal{G'}}[\varphi_{i \neq j}(X, Y)] - 16 \cdot N^{-2} \geq \varepsilon^2 \qedhere
  \end{equation*}
\end{proof}

\subsection{Classical Hardness} \label{sec::hardness}
Let $L_{1,2}(F) = \sum_{S \subseteq [2N], |S| = 2} |\hat{F}(S)|$. The results in this section combined with the tail bound from~\cite{Tal17}, suffices to prove a lower bound on the size of any $\mathsf{AC^0}$ circuit that distinguishes $\mathcal{D}$ from uniform.

\begin{theorem}[\protect{\cite[Claim ~A.5]{DBLP:conf/innovations/ChattopadhyayHL19}}]\label{th::res}
  Let $f$ be a multi-linear function on $\mathbb{R}^N$ and $x \in [-1/2, 1/2]^N$. There exists a distribution over random restrictions $\mathcal{R}_x$ such that for any $y \in \mathbb{R}^N$:
  $$f(x+y) - f(x) = \E_{\rho \sim \mathcal{R}_x}[f_\rho(2\cdot y) - f_\rho(0)] $$
\end{theorem}

\begin{theorem} \label{th::ch}
  Let $t = N^{c-1}$ for a sufficiently large universal constant $c$ (e.g., $c = 40$).
  Let $(X, Y) \in [-1/2, 1/2]^{2N}$, $\Delta X \sim \Nn(0, \varepsilon/t)$ and $\Delta Y = H \cdot \Delta X$. Let $F: \mathbb{R}^{2N} \rightarrow \mathbb{R}$ be a multilinear function where for all random restrictions $L_{1,2}(F_\rho) \leq \ell$:  
  $$(*) = \Big|\E_{\Delta X}\Big[F\big(X + \Delta X, (Y+\Delta Y)^2 - \tfrac{i\cdot \varepsilon}{t}\big)\Big] - \E\Big[F\big(X, Y^2 - \tfrac{(i-1)\cdot \varepsilon}{t}\big)\Big]\Big| \leq \frac{O(\varepsilon \cdot \ell)}{t \sqrt{N}}$$ 
\end{theorem}
\begin{proof}
 Let $Y' = Y^2 - \frac{\varepsilon \cdot (i-1)}{t}$. We can rewrite $(*)$ as follows:
 $$(*) = \E_{\Delta X}\Big[F(X + \Delta X, Y' + 2Y\cdot \Delta Y + \Delta Y^2 - \tfrac{\varepsilon}{t}) - F(X, Y')\Big] $$
 Note that $Y' \in [-1/2, 1/2]^N$. Using Theorem \ref{th::res} we can write this value using random restrictions of $F$:
 $$(*) = \E_{\Delta X, \rho}\Big[F_\rho \big (2 \cdot \Delta X, 2 \cdot \Delta Y'\big) - F_\rho(0, 0)\Big] $$
 Where $\Delta Y' = 2Y\cdot \Delta Y + \Delta Y^2 - \tfrac{\varepsilon}{t}$. Let $Z = (2\Delta X, 2\Delta Y')$. Next we can expand $(*)$ in terms of multilinear expansion of $F_\rho$. From the Isserlis' theorem~\cite{isserlis1918formula} we can see that monomials of degree $k>2$ in the variables $Z$ scale like $O(\frac{k^{k}}{t^{k/2}})$, and there are at most $N^{k}$ such monomials.
 Since we can choose $t$ to be an arbitrarily large power of $N$, their contribution would be negligible, i.e., $o(\eps/t\sqrt{N})$. 
 Furthermore, since for every input variable $X_i$ of $F_\rho$, $\E[Z_i] = 0$ (using the fact that $\E[\Delta Y^2-\tfrac{\eps}{t}] = 0$) we can also ignore monomials of degree $1$. Note that this is precisely why we also break $\varepsilon$ in the telescopic sum. The remaining step is to bound monomials of degree $2$.
 
To account for monomials of degree~$2$, we upper bound the covariances of all pairs of variables from $Z = (2\Delta X, 2\Delta Y')$ by $O(\frac{\eps}{t\sqrt{N}})$.
First, for $i\neq j$ we have that $(\Delta X)_i$ and $(\Delta X)_j$ are independent and thus have covariance $0$.
Second, we bound $\left|\Cov\left(2\left(\Delta X\right)_i, 2\left(\Delta Y'\right)_j \right)\right|$ for any $(i,j)$ by:
 $$ |4 \cdot 2Y_j \Cov ((\Delta X)_i , (\Delta Y)_j) + 4\cdot \Cov((\Delta X)_i, (\Delta Y)_j^2)| \leq \frac{4\varepsilon}{t \sqrt{N}} + O\left(\tfrac{1}{t^{3/2}}\right) = O\left( \frac{\varepsilon}{t\sqrt{N}}\right) $$
 Similarly, we can bound $\left|\Cov \left( 2(\Delta Y')_i, 2(\Delta Y')_j \right) \right|$ for $i\neq j$
 $$4 \cdot 4Y_iY_j \Cov \left((\Delta Y)_i, (\Delta Y)_j \right) + O\left(\frac{1}{t^{3/2}}\right) \leq O\left(\frac{\varepsilon}{t \sqrt{N}}\right)$$
 Finally, we can bound $(*)$:
 \begin{align*}
\left|\E_{\Delta X, \rho}[F_\rho(Z) - F_{\rho}(0^{2N})]\right|&\leq \sum_{i,j} \left|\E_{\Delta X, \rho}[\hat{F_\rho}(\{i,j\}) \cdot  Z_i Z_j]\right| + o\left(\tfrac{\eps \ell}{t\sqrt{N}}\right)\\
&\le \max_{i,j} \left|\Cov(Z_i, Z_j)\right| \cdot \max_\rho  \sum_{\substack{S \subseteq [2N],\\ |S| = 2}} |\hat{F_{\rho}}(S) | + o\left(\tfrac{\eps \ell}{t\sqrt{N}}\right)\\
&\le O\left(\frac{\eps \ell}{t\sqrt{N}}\right).\qedhere
\end{align*}
\end{proof}

\begin{theorem} \label{th::hardness} Let $\ell\ge 1$.
  Let $F: \{\pm1\}^{2N} \rightarrow \{\pm1 \}$ be a multilinear function with bounded spectral norm $L_{1,2}(F) \leq \ell$. Then:
  $$(*) = \E_{(X, Y') \sim \Dd}\Big[F(X, Y') - F(0, 0) \Big] \leq O\left(\frac{\varepsilon \cdot \ell}{\sqrt{N}}\right) $$
\end{theorem}

\begin{proof}
  From Equation \ref{eq1}, it suffices to show:
  $$\E_{(X, Y) \sim \mathcal{G}}\Big[F\big(\trunc(X), \trunc(Y^2 - \varepsilon)\big) - F(0, 0) \Big] \leq O\left(\frac{\varepsilon \cdot \ell}{\sqrt{N}}\right) $$
  We start by writing this value as a telescoping sum. Let $t  = N^{c-1}$ for $c$ the constant guaranteed by \Cref{th::ch}. Set $C$ in the definition of $\eps$ to guarantee that the error in \Cref{th::5.3} is at most $8\cdot N^{-c}$ and that $e^{-1/(8\eps)} \le N^{-(c+1)}$.
  Then we can write:
  $$(*) = \sum_{i = 1}^t \E\Big[F(\trunc(X^{(i)}), \trunc(Y^{(i)^2} - \tfrac{\varepsilon \cdot i}{t})) - F( \trunc(X^{(i-1)}), \trunc(Y^{(i-1)^2} - \tfrac{\varepsilon \cdot (i-1)}{t})) \Big] $$
  where $X^{(i)} = X^{(i-1)} + \Delta X$, $ Y^{(i)} = Y^{(i-1)} + \Delta Y$ and $\Delta X \sim \Nn(0, \varepsilon / t)$ and $\Delta Y = H \cdot \Delta X$. We are going to use Theorem 8 to bound each term in the sum to get the final result.
  Let $E_{i-1}$ be the event that $(X^{(i-1)}, Y^{(i-1)} - \tfrac{\varepsilon \cdot (i-1)}{t}) \in [-1/2, 1/2]^{2N}$. We can prove that $E_{i-1}$ happens with high probability:
  $$\Pr[ (X^{(i-1)}, Y^{(i-1)}) \not\in [-1/2, 1/2]^{(2N)}] \leq 2N\cdot\Pr[|\Nn(0, \eps)| \geq 1/2] \leq 2N \cdot 2 e^{-1/(8\varepsilon)} \leq 4N^{-c} $$
  It is also clear that if $ (X^{(i-1)}, Y^{(i-1)}) \in [-1/2, 1/2]^{2N}$ then $ (X^{(i-1)}, Y^{(i-1)^2} - \tfrac{\varepsilon \cdot (i-1)}{t}) \in [-1/2, 1/2]^{2N}$. So $\Pr[E_{i-1}] \geq 1-4N^{-c}$. Next, we just have to use \Cref{th::5.3,th::ch} combined with triangle inequality to prove the desired claim. We can bound the $i$-th term in the sum, $S_i$, by:
  \begin{align*}  
    S_i &\leq \E\Big[F(X^{(i)}, Y^{(i)^2} - \tfrac{\varepsilon \cdot i}{t}) - F( X^{(i-1)}, Y^{(i-1)^2} - \tfrac{\varepsilon \cdot (i-1)}{t}) \Big| E_{i-1}\Big] \\
    &+ \E\Big[F(\trunc(X^{(i)}), \trunc(Y^{(i)^2} - \tfrac{\varepsilon \cdot i}{t})) - F(X^{(i)}, Y^{(i)^2} - \tfrac{\varepsilon \cdot i}{t}) \Big| E_{i-1}\Big] \\
    &+ 2 \Pr[\lnot E_{i-1}] \leq O\left(\frac{\varepsilon \cdot \ell}{t \sqrt{N}}\right) + O(N^{-c})
  \end{align*}
  So we can bound the sum $(*) \leq O(\frac{\varepsilon \cdot \ell}{\sqrt{N}}) + O(t N^{-c}) \leq O(\frac{\varepsilon \cdot \ell}{\sqrt{N}})$.
\end{proof}

\begin{remark} \label{rem:llh}
Note that this proof can easily be extended to the ``long list'' version of the problem, where we are given oracle access to a list of pairs of functions $(f_i, g_i)$ rather than a single pair that has been shown. We have already shown this for Theorem \ref{th::5.2} and \ref{th::5.3} in Appendix \ref{app:trunc}, which directly proves Theorem \ref{th::ch}. The only remaining step that we need to verify in Theorem \ref{th::hardness} is the probability of the event $E_i$. Similar to what has been shown in Appendix \ref{app:trunc}, we can choose $\varepsilon$ to be small enough so that $p(N) \cdot e^{-1/(8\varepsilon)}$ remains exponentially small, and the rest of the proof follows immediately. Later we will use this result to give evidence for $\mathsf{LLQSV}$.
\end{remark}

\subsection{More details on the relation between \texorpdfstring{$\sqfor$}{Squared Forrelation} and Long List} \label{subsec:llf}
As mentioned before, our goal is to use samples $(f, g) \sim \mathcal{D}$, and then use $g$ to get samples according to the Fourier spectrum. However, if we choose a uniformly random sample such that $g(x)=1$, we would not get a sample \emph{exactly} according to the Fourier distribution of $f$. Intuitively, this is a ``flattened version'' of the Fourier transform distribution of $f$, but we can still prove a $\mathsf{PH}$ lower bound for this variant of $\bllqsv$, where the samples are taken according to this flattened distribution. We show that this distinction does not matter, by first noting that we can define a modified $\mathsf{HOG}$ score -- Rejection Sampling $\mathsf{HOG}$ ($\mathsf{RHOG}$) -- so that $\mathsf{RHOG}$ and the flattened distribution have the same relation as $\mathsf{HOG}$ and $\bllqsv$. Furthermore, we prove that a quantum algorithm can still solve both $\mathsf{HOG}$ and $\mathsf{RHOG}$ by sampling according to the Fourier transform distribution of $f$.
Let $\mathcal{D}$ be the distribution defined in the previous section, and let $\mathcal{D}_f$ be the marginal distribution over $g$ for a fixed $f$. Given a random Boolean function $f$, $\mathcal{R}_f$ samples $s \in \{0, 1\}^n$ as follows:
  \begin{itemize}
  \item Sample $g \sim \mathcal{D}_f$.
    \item Sample $x \in \{0, 1\}^n$ uniformly at random and output $x$ if $g(x) = 1$.
    \item Output a random element after $4n^2$ unsuccessful attempts.
  \end{itemize}
Suppose we are given access to oracle $\mathcal{O}$ which is a long (exponentially large) list of random Boolean functions $f_1, f_2, \ldots, f_T: \{\pm 1\}^n \rightarrow \{0, 1\}$ paired with strings $s_1, s_2, \ldots, s_T$. Here we define $\bllqsv$ to be the problem of distinguishing samples from $\mathcal{R}_f$ from uniform. One can hope that since $g$ is correlated with squared Fourier coefficients of $f$, these samples are also close to the ``Fourier distribution''. In fact, we later show that on average (over the random Boolean function and the samples), $\mathcal{R}_f$ generates heavy Fourier coefficients. The same argument also works to prove that an honest quantum sampler, that samples according to the Fourier coefficients, passes the following score:

\begin{problem} $(\mathsf{RHOG})$\label{p::rejection}
  Let $b=1+\poly(1/n)$. Given a random Boolean function $f: \{0,1\}^N \rightarrow \{\pm1\}$, output an outcome $s \in \{0, 1\}^N$ such that:
  $$\E_{f,s} \left[\Pr\left[\mathcal{R}_f \text{~samples~} s\right]\right] \geq \frac{b}{N}$$
\end{problem}
One can easily verify that the existence of a deterministic quantum algorithm for solving $\mathsf{RHOG}$ also gives a $\mathsf{QCAM}$ algorithm for (modified) $\bllqsv$. This follows from a collision finding algorithm similar to the original $\mathsf{LLQSV}$ argument (see Appendix \ref{app:a-certified}). Here we prove that first, the usual quantum algorithm that samples according to the Fourier distribution passes solves $\mathsf{RHOG}$. Next, we show that $\sqfor$ can be reduced to $\bllqsv$, which implies that $\bllqsv$ is not in $\mathsf{PH}$.

We prove the first result by using two properties of the distribution $\mathcal{G'}$. The first directly follows from a tail bound of Chi-square variables:
\begin{lemma}[\cite{lm2000}]\label{lemma::lm}
Let $(Y_1, \ldots, Y_N) \sim \mathcal{N}(0, 1)$ and let $a_1, \ldots, a_N \geq 0$. Let 
$$Z = \sum_i a_i (Y_i^2 - 1) $$
Then for any $\Delta > 0$ we have:
\begin{itemize}
\item $\Pr\big[Z \geq 2 |a|_2 \sqrt{\Delta} + 2|a|_\infty \Delta \big] \leq \exp(-\Delta)$
\item $\Pr\big[Z \leq -2 |a|_2 \sqrt{\Delta} \big] \leq \exp(-\Delta)$
\end{itemize}
\end{lemma}
\begin{theorem} \label{th::lm}
Let $(X, Y) \sim \mathcal{G'}$, then $\Pr \big[ \big|\sum_i Y_i \big| \geq 3 \sqrt{N} \big] \leq 2\exp(-1/\varepsilon)  $.
\end{theorem}
\begin{proof}
Using Lemma \ref{lemma::lm} by letting $a = (\varepsilon, \ldots, \varepsilon)$ and $\Delta = \frac{1}{\varepsilon}$ we get:
\begin{itemize}
\item $\Pr\big[Z \geq  2 \sqrt{N} + 2 \big] \leq \exp(-\frac{1}{\varepsilon})$
\item $\Pr\big[Z \leq -2 \sqrt{N} \big] \leq \exp(-\frac{1}{\varepsilon})$
\end{itemize}
And a union bound gives us the claim.
\end{proof}
Recall that for $(X, Y) \in \mathcal{G'}$ with high probability we have $\trunc(X, Y) = (X, Y)$:
\begin{theorem} \label{th::trnc}
Let $(X, Y) \sim \mathcal{G}'$, then $\Pr\big[\trunc(X, Y) \neq (X, Y)\big] \leq 2N^{-2}$.
\end{theorem}

\begin{corollary}
Let $(X, Y) \sim \mathcal{D}$ and $\delta = N^{-1/3}$, then with probability at least $1 - 5N^{-2}$ we have:
$$ (1-\delta)\frac{N}{2} \leq \big|\{i | Y_i = 1\} \big| \leq (1+\delta) \frac{N}{2} $$
\end{corollary}
\begin{proof}
This follows from a Chernoff bound applied to Theorem \ref{th::lm} and \ref{th::trnc}. Given $(X', Y') \sim \mathcal{G}'$, by a union bound we get that the probability that both requirements $\trunc(X', Y') = (X', Y')$ and $|\sum Y'_i| < 3\sqrt{N}$ are satisfied is at least $1 - 4N^{-2}$.
In such a case, recall that $Y$ is attained from $Y'$ by randomized rounding (without truncations since $Y' = \trunc(Y')$.
Let $\delta' = \delta/2$. Then by Chernoff bound, with probability $1 - \exp(\Omega(\delta^2 N))\ge 1-N^{-2}$:
  $$\left|\{i | Y_i = 1 \}\right| \leq (1+\delta') \left(\frac{N}{2} + 3\sqrt{N}\right) \leq (1+\delta)\frac{N}{2} $$
  and,
  $$ \left|\{i | Y_i = 1 \}\right| \geq (1-\delta') \left(\frac{N}{2} - 3\sqrt{N}\right) \geq  (1-\delta) \frac{N}{2}$$
   for large enough $N$. This means that the claim holds over $\mathcal{D}$ with probability at least $1 - 5N^{-2}$.
\end{proof}
Another direct consequence of this property is that we can bound the success probability of getting a non-uniform sample after $4n^2$ steps in rejection sampling. Note that the probability of the rejection sampling failing is bounded by $\big((1+\delta)/2\big)^{4n^2}$. For large enough $N$ and $\delta = 1/N^{1/3}$, we know $((1+\delta)/2)^2 \leq 1/2$, and thus
$$((1+\delta)/2)^{4n^2} \leq 2^{-2n^2} \ll N^{-2} $$
Now we can prove that given random Boolean function $f: \{\pm 1\}^N \rightarrow \{0, 1\}$, sampling according to the Fourier coefficients of $f$ solves $\mathsf{RHOG}$. Define $\mathcal{C}_f$ to be the Fourier transform distribution over outcomes according to $f$, i.e., $\mathcal{C}_f$ samples $x$ with probability $\hat{f}(x)^2$. 
\begin{claim}
  $$\E_{\substack{f,\\x\sim \mathcal{C}_f}}\left[\Pr\left[\mathcal{R}_f \text{~samples~}x\right]\right] 
  \geq \frac{1+\poly(\varepsilon)}{N} $$
\end{claim}
\begin{proof}
   From Theorem \ref{th::qa}, we know that $ \E_\mathcal{D} \big[\sum_{g(x) = 1} \hat{f}(x)^2\big] \geq \frac{1}{2} + \frac{\varepsilon^2}{4}$. 
Let $E$ be the event that :
\begin{enumerate}
  \item  $ (1-\delta)\frac{N}{2} \leq \big|\{i | Y_i = 1\} \big| \leq (1+\delta) \frac{N}{2} $
  \item The rejection sampling algorithm outputs an $s$ such that $g(s) = 1$. 
\end{enumerate}
 where $\delta  = N^{-1/3}$. 
 We know that $\Pr[E] \geq 1-10/N^2$ and $ \E_\mathcal{D} \big[\sum_{g(x) = 1} \hat{f}(x)^2 | E\big] \geq \frac{1}{2} + \frac{\varepsilon^2}{8}$. 
   We can  lower bound  $$\E_{\substack{f,\\x\sim \mathcal{C}_f}}\left[\Pr\left[\mathcal{R}_f \text{~samples~}x\right]\right] \ge \Pr[E]\cdot \Pr_{\substack{f,\\x\sim \mathcal{C}_f}}\left[\mathcal{R}_f \text{~samples~}x|E\right].$$
   Note further that conditioned on $E$, the rejection sampling algorithm draws a random element from the set $\{y:g(y) = 1\}$. Thus we may rewrite
  \begin{align*}
  \Pr[E]\cdot \Pr_{\substack{f,\\x\sim \mathcal{C}_f}}\left[\mathcal{R}_f \text{~samples~}x|E\right]
   &= 
   \Pr[E] \cdot \E_{(f, g) \sim \mathcal{D}} \bigg[\frac{\sum_{g(x) = 1} \hat{f}(x)^2}{ |\{y:g(y) = 1\}| } \;\bigg|\; E \bigg] \\
&\ge\Pr[E]\cdot\frac{1/2 + \varepsilon^2/8}{(1+\delta)N/2} \\
&\ge \Pr[E] \cdot \frac{(1-\delta)(1+ \eps^2/4)}{N} \end{align*}
  which is at least $\frac{1+\varepsilon^2/8}{N}$ for $N$ large enough (where we used the fact that $\eps = \Theta(1/\log N)$ is much larger than $10/N^2$ and $\delta$).
\end{proof}
The only remaining step is to prove the hardness of (modified) $\bllqsv$. First, we define the following problem:
\begin{problem} \label{prob:llf+g}
  Given oracle access to a length $T = 2^{3n}$ list of triplets $(f_i, g_i, s_i)$ where $f_i, g_i: \{\pm 1\}^N \rightarrow \{0, 1\}$ and $s_i \in \{0, 1\}^N$ is sampled according to the rejection sampling algorithm applied to $g_i$ distinguish between the following two cases:
  \begin{enumerate}
    \item For every $i$, $f_i, g_i$ are chosen uniformly at random.
    \item For every $i$. $f_i, g_i$ are sampled according to $\mathcal{D}$.
  \end{enumerate}
\end{problem}
We can use this to prove our final result:
\begin{theorem} \label{th:ac02}
No constant depth Boolean circuit of size $\mathrm{quasipoly}(N)$ can solve $\bllqsv$ with advantage more than $\mathrm{polylog}(N)/\sqrt{N}$.
\end{theorem}
\begin{proof}
 First, it is clear that we can reduce Problem \ref{prob:llf+g} to $\bllqsv$ since we have access to the list $(f_i, s_i)$. To prove the same oracle separation result using Problem \ref{prob:llf+g}, we first notice that we can prove the same $\mathsf{AC^0}$ lower bound if we only had access to $(f_i, g_i)$, but not $s_i$ (Remark \ref{rem:llh}). Furthermore, by choosing $\varepsilon$ in the definition of $\mathcal{D}$ small enough, the lower bound even holds for a list of length $2\cdot2^{n^3}$. The final step is to prove that $s_i$'s do not give us extra computation power. Suppose there exists a circuit $C$ of size $s$ and depth $d$ for Problem \ref{prob:llf+g}, we can convert this circuit to work without having access to $s_i$. First, we can replace $s_i$ with a small circuit that given $r_i \in \{0, 1\}^{\poly(n)}$, the random seed of the rejection sampling, simulate rejection sampling, making $O(n^2)$ queries to $g_i$ and outputs $s_i$.  Such a computation can be described by a decision tree on that reads $r_i$ and at most $O(n^2)$ entries of $g_i$, and thus also by a DNF of size $2^{\poly(n)} =  \mathrm{quasipoly}(N)$. Making all these replacements everywhere gives a new circuit with size $s' = s\cdot \mathrm{quasipoly}(N)$ (which is still quasi-polynomial in $N$) and depth $d+2$. Furthermore, if we have access to a list of length $2\cdot2^{n^3}$ of pairs of functions, we can use the second half as the random seed, instead of accessing $r_i$ directly. So this new circuit of size $s'$ and depth $d+2$ must also distinguish the two distributions with the same advantage, having only access to a list of length $2\cdot 2^{n^3}$ of $(f_i, g_i)$. Thus, the same lower bound applies to the size of the starting circuit $s$.
 \end{proof}

\section*{Acknowledgements}

We thank Scott Aaronson and Umesh Vazirani for helpful discussions.
A.B. was supported in part by the AFOSR under grant FA9550-21-1-0392.
B.F. and R.B. acknowledge support from AFOSR (YIP number FA9550-18-1-0148 and FA9550-21-1-0008).
This material is based upon work partially
supported by the National Science Foundation under Grant CCF-2044923 (CAREER) and by the U.S. Department of Energy, Office of Science, National Quantum Information Science Research Centers (Q-NEXT) as well as by DOE QuantISED grant DE-SC0020360.
A.T. was supported in part by the National Science Foundation under Grant CCF-2145474 (CAREER) and by a Sloan Research Fellowship.

\nocite{*}
\bibliography{certifiedbib}
\bibliographystyle{alpha-mod}

\appendix
\section{Truncated Gaussians - Proofs} \label{app:trunc}
Here we prove a long list version of Theorems \ref{th::5.2} and \ref{th::5.3}. The core statements are the same, but the expectation is over the distribution $\mathcal{G'}^{p(N)}$ for some polynomial $p$. In LLQSV we are interested in the case where $p(N) = N^3$. We show that for any polynomial $p$ we can always choose $\varepsilon = 1/(C\ln N)$, to make expectation values in the next theorems arbitrarily small ($N^{-c}$ for some constant $c$ that depends on $p$ and $C$).
\begin{proof} [Proof of Thm. \ref{th::5.2}]
We can write:
\begin{align*}
(*)&=\E_{(x,y) \sim \mathcal{G'}^{p(n)}}\left[\prod_{i=1}^{p(N)\cdot N} \max \left(1,\left|x_{i}\right|\right) \cdot \prod_{i=1}^{p(N)\cdot N} \max \left(1,\left|y_{i}\right|\right) \cdot \mathbbm{1}_{(x, y) \neq \operatorname{trnc}(x, y)}\right] \\
& = \E_{(x,y) \sim \mathcal{G'}^{p(n)}}\left[\prod_{i=1}^{p(N)\cdot N} \max \left(1,\left|x_{i}\right|\right) \cdot \prod_{i=1}^{p(N)\cdot N} \max \left(1,\left|y_{i}\right|\right)\right] \\
& - \E_{(x,y) \sim \mathcal{G'}^{p(n)}}\left[\prod_{i=1}^{p(N)\cdot N} \max \left(1,\left|x_{i}\right|\right) \cdot \prod_{i=1}^{p(N)\cdot N} \max \left(1,\left|y_{i}\right|\right) \cdot \mathbbm{1}_{(x, y) = \operatorname{trnc}(x, y)}\right] \\
& = \E_{(x,y) \sim \mathcal{G'}^{p(n)}}\left[\prod_{i=1}^{p(N)\cdot N} \max \left(1,\left|x_{i}\right|\right) \cdot \prod_{i=1}^{p(N)\cdot N} \max \left(1,\left|y_{i}\right|\right)\right] - \Pr\left[(x, y) = \trnc(x, y) \right]
\end{align*}
Let $z \coloneqq (x, y)$ and $\varepsilon = \frac{1}{2k\ln(p(N)\cdot N)}$ for some constant $k$. We first bound the second term:
$$\Pr\big[z = \trnc(z) \big] \geq 1 - \sum_i \Pr[z \neq \trnc(z)] = 1 - 2N\cdot p(N) \cdot e^{-1/(2\varepsilon)}\geq 1 - N^{-k+1}$$
Next we can bound the first term using Cauchy-Schwarz:
\begin{align*} 
& \leq \sqrt{\E_{x}\left[\prod_{i=1}^{p(N)\cdot N} \max \left(1,x_{i}^2\right)\right]} \cdot \sqrt{\mathbb{E}_{x}\left[\prod_{i=1}^{p(N)\cdot N} \max \left(1,(x_{i}^2-\varepsilon)^2\right)\right]} \\
& = \E_{x}\left[ \max \left(1,x_{i}^2\right)\right]^{p(N)\cdot N/2} \cdot \E_{x}\left[ \max \left(1,(x_{i}^2-\varepsilon)^2\right)\right]^{p(N)\cdot N/2} 
\end{align*}
We can write the first expectation value as:
\begin{align*}
    2 \int_0^1 \frac{1}{\sqrt{2\pi\varepsilon}} e^{-x^2/2\varepsilon} dx + 2 \int_1^\infty \frac{x^2}{\sqrt{2\pi\varepsilon}}e^{-x^2/2\varepsilon}dx &\leq 1 + \varepsilon \cdot \erfc(1/\sqrt{2\varepsilon}) + \sqrt{2\varepsilon/\pi} \cdot e^{-1/2\varepsilon} \\ 
    & \leq 1 + \varepsilon e^{-1/2\varepsilon} + \sqrt{2\varepsilon/\pi}e^{-1/2\varepsilon} \\
& \leq 1+N^{-k+1}\cdot p(N)^{-k}
\end{align*}
For the second expectation value:
\begin{align*}
    &2 \int_0^{\sqrt{1+\varepsilon}} \frac{1}{\sqrt{2\pi\varepsilon}} e^{-x^2/2\varepsilon} dx + 2 \int_{\sqrt{1+\varepsilon}}^\infty \frac{(x^2-\varepsilon)^2}{\sqrt{2\pi\varepsilon}} e^{-x^2/2\varepsilon} dx \\
    &\leq 2\int_0^{\sqrt{1+\varepsilon}} \frac{1}{\sqrt{2\pi\varepsilon}} e^{-x^2/2\varepsilon} dx +  2 \int_{1}^\infty \frac{x^4}{\sqrt{2\pi\varepsilon}} e^{-x^2/2\varepsilon} \\
    &\leq 1 + 3\varepsilon^2 \erfc(1/\sqrt{2 \varepsilon}) + \frac{\varepsilon(1+3\varepsilon)}{\sqrt{2\pi \varepsilon}}e^{-1/2\varepsilon} \\
    & \leq 1 + 3\varepsilon^2 e^{1/2\varepsilon} + \frac{\varepsilon(1+3\varepsilon)}{\sqrt{2\pi \varepsilon}}e^{-1/2\varepsilon} \leq   1+N^{-k+1}\cdot p(N)^{-k}
\end{align*}
So we can bound the first term by  $(1+N^{-k+1}\cdot p(N)^{-k})^{N\cdot p(N)} \leq 1+N^{-k+2}$. Combining both we get the claim:
$$(*) \leq 1+N^{-k+2} - (1-N^{-k+1}) \leq 2 N^{-k+2} $$

\end{proof}

\begin{proof} [Proof of Thm. \ref{th::5.3}]
  The original proof works without any change with the modified Claim 5.2. 
  \begin{align*}
&\underset{z \sim \mathcal{G'}^{p(N)}}{\E}\left[\left|F\left(\operatorname{trnc}\left(z_{0}+p \cdot z\right)\right)-F\left(z_{0}+p \cdot z\right)\right|\right] \\& \leq \underset{z \sim \mathcal{G'}^{p(N)}}{\E}\left[\left(1+\left|F\left(z_{0}+p \cdot z\right)\right|\right) \cdot \mathbbm{1}_{\trnc(z_0+p\cdot z) \neq z_0+p\cdot z}\right] \\
& \leq \underset{z \sim \mathcal{G'}^{p(N)}}{\E}\left[\left(1+\left|F\left(z_{0}+p \cdot z\right)\right|\right) \cdot \mathbbm{1}_{z \neq \operatorname{trnc}(z)}\right] \\
& \leq \underset{z \sim \mathcal{G'}^{p(N)}}{\E}\left[\left(1+\prod_{i=1}^{2 N \cdot p(N)} \max \left(1,\left|\left(z_{0}\right)_{i}+p \cdot z_{i}\right|\right)\right) \cdot \mathbbm{1}_{z \neq \operatorname{trnc}(z)}\right] \\
& \leq \underset{z \sim \mathcal{G'}^{p(N)}}{\E}\left[2 \cdot \prod_{i=1}^{2 N \cdot p(N)} \max \left(1,\left|\left(z_{0}\right)_{i}+p \cdot z_{i}\right|\right) \cdot \mathbbm{1}_{z \neq \operatorname{trnc}(z)}\right]
\end{align*}
And since we know $\prod_{i=1}^{2 N \cdot p(N)} \max \left(1,\left|\left(z_{0}\right)_{i}+p \cdot z_{i}\right|\right) \leq \prod_{i=1}^{2 N\cdot p(N)} \max \left(1, p_{0}+p\left|z_{i}\right|\right) \leq \prod_{i=1}^{2 N\cdot p(N)} \max \left(1,\left|z_{i}\right|\right)$, we can write:
\begin{align*}
&\underset{z \sim \mathcal{G'}^{p(N)}}{\E}\left[\left|F\left(\operatorname{trnc}\left(z_{0}+p \cdot z\right)\right)-F\left(z_{0}+p \cdot z\right)\right|\right] \\
&\quad \leq \underset{z \sim \mathcal{G'}^{p(N)}}{\E}\left[2 \cdot \prod_{i=1}^{2 N\cdot p(N)} \max \left(1,\left|z_{i}\right|\right) \cdot \mathbbm{1}_{z \neq \operatorname{trnc}(z)}\right] 
\leq 4 \cdot N^{-k+2}.\qedhere
\end{align*}
\end{proof}

\section{Aaronson's Certified Random Number Generation} \label{app:a-certified}
  For completeness, we summarize how certified random number generation is possible from ``long list'' conjectures. Consider the following problem:

\begin{problem} \label{prob::xeb}
  Given a random quantum circuit $C$ on $n$ qubits, output a sample $s$ such that:
  $$\E_C\Big[P_C(s)\Big] \geq \frac{b}{N} $$
  where $P_C(s) = |\bra{s}C\ket{0^n}|^2$ and $b = 1+\varepsilon$ and $\varepsilon = \text{poly}(1/n) $.
\end{problem}
 
 We already know that the trivial quantum algorithm that samples from the output distribution of $C$ is a solution for this problem \cite{PhysRev.104.483, boixo-characterizing}, and in fact, we have strong evidence that current NISQ devices are capable of solving this problem \cite{boixo-supremacy}. If we can prove that any algorithm in $\mathsf{BQP}$ for Problem \ref{prob::xeb} has min-entropy larger than the random seed required by our pseudorandom generator, we can connect cryptographic assumptions to certified random bit generation. Our main tool for proving such a claim is the hardness of $\mathsf{LLQSV}$. Suppose we are given access to oracle $\mathcal{O}$ which is a long (exponentially large) list of $n$ qubit random quantum circuits $C_1, C_2, \ldots, C_T$ paired with strings $s_1, s_2, \ldots, s_T$. We can informally define the Long List Hardness Assumption ($\mathsf{LLHA}$):
\begin{conjecture} $\mathsf{(LLHA)}$ \label{conj::LLQSV}
  Given oracle access to $\mathcal{O}$ it is $\mathsf{QCAM}$-Hard to distinguish the following two cases:
  \begin{enumerate}
    \item Each $s_i$ is sampled uniformly at random.
    \item Each $s_i$ is sampled from the output distribution of $C_i$.
  \end{enumerate}
\end{conjecture}
\noindent To summarize the overall proof strategy, we focus on refuting the existence of a deterministic $\mathsf{BQP}$ algorithm for Problem \ref{prob::xeb}. The key idea is that if we have access to a deterministic function solving Problem \ref{prob::xeb}, we can use this function in an approximate counting argument to count the number of collisions between the output of that function and the output of the quantum circuit. Furthermore, if the list in the conjecture is long enough, we can use Chernoff bound to separate the two cases which gives a $\mathsf{QCAM}$ protocol for $\mathsf{LLQSV}$. Suppose we are given an instance of $\mathsf{LLQSV}$, where $T = N^3 = 2^{3n}$. Let $S = \sum_i^T P_{C_i}(s_i)$. Then by Hoeffding's inequality, we have:
$$\Pr \Big[S \leq bN^2 - \frac{\varepsilon}{2}N^2\Big] \geq \exp (- \frac{\varepsilon^2 N^4}{N^3}) \geq \exp (-O(n)) $$
Let $V = |\{i | s_i \in Q(C_i) \} | $. Next, we can use Chernoff bound to separate the two cases:
\begin{enumerate}
  \item $\Pr \Big[V \geq (1+\delta)N^2 \Big] \leq \exp (-\frac{\delta^2}{2+\delta}N^2) \leq \exp(-N) $
  \item $\Pr \Big[V \leq (1-\delta)((1+\frac{\varepsilon}{2})N^2) \Big] \leq \exp(-\frac{\delta}{2}(1+\frac{\varepsilon}{2})N^2) \leq \exp(-N) $
\end{enumerate}
By setting $\delta = \varepsilon^2$, with high probability:
\begin{enumerate}
  \item In the uniform case $V < (1+\varepsilon^2)N^2$.
  \item Otherwise $V > (1+\frac{\varepsilon}{4})N^2$
\end{enumerate}
Lastly, we can use approximate counting to estimate this value up to a multiplicative factor of $(1+\varepsilon^2)$ to decide between these two cases with high probability. In \cite{a-talk}, Aaronson provides further evidence that this argument can be extended to constant min-entropy algorithms and provides methods to accumulate entropy.

\section{Proof of Claim~\ref{claim:st_dist}} \label{ap:st_dist}
We restate the claim.
\begin{claim}
Conditioned on $\mathcal{E}$, with $1-\exp(-\Omega(n))$ probability over the choice of $\tau$, the distributions $p_{Z|\mathcal{P}}$ and $\tilde{p}_{Z|\mathcal{P}}$ are $\exp(-\Omega(n))$-close in statistical distance.
\end{claim}
\begin{proof}
We assume that $F$ is good in the sense that Eq.\eqref{eq:tau} holds. This means that $Z: p_Z \ge \tau$ and $Z: \tilde{p}_Z \ge \tau$ capture at least $\gamma/\poly(n) \ge 2^{-o(n)}$ of the probability mass of the two distributions $p$ and $\tilde{p}$.

Recall that $\tau$ is a uniformly random integer multiple of $2^{-n/20}$ in the range $[2^{-n/30}, 2\cdot 2^{-n/30}]$, and there are $2^{-n/30 + n/20}+1 = 2^{n/60}+1$ such choices.

Consider the $2^{n/60}$ choices of $\tau$. They partition the values of $p_Z$ into buckets. 
For $i= 0,1, \ldots,  2^{n/60}-1$, the $i$-th bucket contain all $Z: p_Z \in [2^{-n/30} + i\cdot 2^{-n/20}, 2^{-n/30} + (i+1)\cdot 2^{-n/20})$.
Finally, let the last bucket contain all $Z: p_Z \ge 2\cdot 2^{-n/30}$.
The key idea is that there cannot be too many buckets whose probability mass is large.
More precisely, the $i$-th bucket is \emph{problematic} if its probability mass is at least $2^{-n/80}$. There cannot be more than $2^{n/80}$ problematic buckets.
Thus, most buckets are non-problematic, and in fact, most consecutive three buckets are non-problematic.
Say we pick $\tau$ to be $2^{-n/30} + i\cdot 2^{-n/20}$ for $i$ such that both $i-1$ and $i$ are non-problematic.

Note that due to event $\mathcal{E}$ each $Z$ would get either the same bucket according to $p_Z$ and $\tilde{p}_Z$ or a bucket off by $1$ (either up or down).
Thus if both $i-1$ and $i$ are non-problematic bucket than picking the threshold $\tau = 2^{-n/30} + i\cdot 2^{-n/20}$ guarantees that $$\sum_{Z: p_Z \ge \tau, \tilde{p}_Z< \tau} p_Z \le \sum_{Z: p_Z \in [\tau, \tau + 2^{-n/20})} p_Z \le 2^{-n/80}$$ and
$$\sum_{Z: p_Z < \tau, \tilde{p}_Z\ge  \tau} p_Z \le
\sum_{Z: p_Z \in [\tau-2^{-n/20}, \tau)} p_Z\le 2^{-n/80}$$
and so the sums $S := \sum_{Z:p_Z \ge \tau } p_Z $ and $\tilde{S}:=\sum_{Z:\tilde{p}_Z \ge \tau } \tilde{p}_Z$ are close since
\begin{align*}
|S-\tilde{S}|&=\left|
\sum_{Z:p_Z \ge \tau } p_Z 
- \sum_{Z:\tilde{p}_Z \ge \tau} \tilde{p}_Z
\right|
\\
&\le 
\left|
\sum_{Z:p_Z \ge \tau } p_Z 
- \sum_{Z:\tilde{p}_Z \ge \tau} p_Z \pm 2^{-n/20}
\right|\\
&\le 
\left|
\sum_{Z:p_Z \ge \tau } p_Z 
- \sum_{Z:\tilde{p}_Z \ge \tau} p_Z
\right| + 2^{-n/60}\\
&\le 2^{-n/80} + 2^{-n/80} + 2^{-n/60} = 2^{-\Omega(n)}\end{align*}
On the other hand both sums are at least $\gamma/\poly(n) = \exp(-o(n))$ and thus their ratio is $1\pm \exp(-\Omega(n))$ and $|\frac{1}{S} - \frac{1}{\tilde{S}}| \le \exp(-\Omega(n))$.
Finally, we have
\begin{align*}&
\sum_{Z}{\left|\frac{1}{S} \cdot 
|p_Z \cdot \mathds{1}_{p_Z \ge \tau}| 
- \frac{1}{\tilde{S}}
\cdot 
|\tilde{p}_Z \cdot \mathds{1}_{\tilde{p}_Z \ge \tau}|\right|}\\
&\le
\sum_{Z:p_Z \ge \tau, \tilde{p}_Z <\tau}{\frac{1}{S} |p_Z|} 
+ \sum_{Z:p_Z < \tau, \tilde{p}_Z \ge \tau}{\frac{1}{\tilde{S}} |\tilde{p}_Z|} 
+ \sum_{Z: p_Z, \tilde{p}_Z \ge \tau}
\left|\frac{1}{S}p_Z - \frac{1}{\tilde{S}}\tilde{p}_Z  \right|
\end{align*}
We already got that the first two sums are $2^{-\Omega(n)}$. As for the third sum, we use triangle inequality to bound it by $\sum_{Z: p_Z, \tilde{p}_Z \ge \tau}
\left|\frac{1}{S}p_Z - \frac{1}{S}\tilde{p}_Z  \right|+\sum_{Z: p_Z, \tilde{p}_Z \ge \tau}
\left|\frac{1}{S} - \frac{1}{S}\right|\tilde{p}_Z$, where both terms are $\exp(-\Omega(n))$ small.
\end{proof}

\end{document}